\documentclass[review,onefignum,onetabnum]{siamart190516}

\usepackage{lipsum}
\usepackage{amsfonts}
\usepackage{graphicx}
\usepackage{epstopdf}
\usepackage{algorithmic}
\ifpdf
  \DeclareGraphicsExtensions{.eps,.pdf,.png,.jpg}
\else
  \DeclareGraphicsExtensions{.eps}
\fi

\usepackage{color}

\usepackage{subcaption}

\newcommand{\vecc}{\mathbf{c}}
\newcommand{\id}{{\mathbb I}}
\newcommand{\tol}{\textnormal{TOL}}
\newcommand{\real}{{\mathbb R}}
\newcommand{\dx}{\textnormal{d}x}
\newcommand{\textn}[1]{\textnormal{#1}}
\newcommand{\pt}{\partial_t}
\newcommand{\dy}{\textnormal{d}y}
\newcommand{\dzeta}{\textnormal{d}\zeta}
\newcommand{\bse}{\begin{subequations}}
\newcommand{\ese}{\end{subequations}}

\DeclareMathOperator*{\argmin}{arg\,min}

\newcommand{\Cbb}{{\mathbb C}}
\newcommand{\Dbb}{{\mathbb D}}

\newcommand{\Ccal}{\mathcal{C}}
\newcommand{\Bcal}{\mathcal{B}}
\newcommand{\Zcal}{\mathcal{Z}}
\newcommand{\Ncal}{\mathcal{N}}
\newcommand{\Lcal}{\mathcal{L}}
\newcommand{\Fcal}{\mathcal{F}}
\newcommand{\Jcal}{\mathcal{J}}
\newcommand{\Rcal}{\mathcal{R}}
\newcommand{\Dcal}{\mathcal{D}}

\newcommand{\sigmabo}{\boldsymbol{\sigma}}
\newcommand{\Sigmabo}{\boldsymbol{\Sigma}}

\newcommand{\vecN}{\mathbf{N}}
\newcommand{\vecX}{\mathbf{X}}

\newcommand{\dnabla}{\nabla \nabla}
\newcommand{\dnablaii}{\nabla_i \nabla_i}
\newcommand{\dnablajj}{\nabla_j \nabla_j}

\newcommand{\vertiii}[1]{{\vert\kern-0.25ex\vert\kern-0.25ex\vert #1 \vert\kern-0.25ex\vert\kern-0.25ex\vert}}
\newcommand{\norm}[1]{\lVert#1\rVert}
\newcommand{\bnorm}[1]{\Big\lVert#1\Big\rVert}

\newcommand{\Hdiv}{\mathbf{H}(\divvv,\Omega)}
\newcommand{\divvv}{{\textnormal{div}}}

\newcommand{\etaRK}{\eta_{\textnormal{R},K}}
\newcommand{\etaDF}{\eta_{\textnormal{DF},K}}
\newcommand{\etaphih}{\eta_{\Phi}}
\newcommand{\etagh}{\eta_{g,K}}
\newcommand{\etarh}{\eta_{r,K}}
\newcommand{\etaosch}{\eta_{f}}
\newcommand{\etaphihj}{\eta_{\Phi,i}}

\newcommand{\etaghj}{\eta_{g,K,i}}
\newcommand{\etarhj}{\eta_{r,K,i}}

\newcommand{\dist}{\textnormal{dist}}

\newsiamremark{remark}{Remark}
\newsiamremark{hypothesis}{Hypothesis}
\crefname{hypothesis}{Hypothesis}{Hypotheses}
\newsiamthm{claim}{Claim}

\newsiamthm{assum}{Assumption}
\newsiamthm{df}{Definition}

\headers{Modeling the process of speciation using a multi-scale framework}{M. K. Brun, E. Ahmed, J. M. Nordbotten, and N. C. Stenseth}

\title{Modeling the process of speciation using a multi-scale framework including a posteriori error estimates\thanks{Submitted to the editors DATE.
\funding{This work was funded 
in part by Norwegian Research Council project no.~263149.}}}

\author{
	Mats K. Brun\thanks{CEES, Dept. of Biosciences, University of Oslo, NO-0316 Oslo, Norway. Current address: Expert Analytics, NO-0179 Oslo, Norway 
  			(\email{m.k.brun@ibv.uio.no}/\email{mats@xal.no}).} 
	\and Elyes Ahmed\thanks{SINTEF, NO-0314 Oslo, Norway 
  			(\email{elyes.ahmed@sintef.no}).}
	\and Jan M. Nordbotten\thanks{Department of Mathematics, University of Bergen, NO-5020 Bergen, Norway 
  			(\email{jan.nordbotten@uib.no}).}
	\and Nils Chr. Stenseth\thanks{CEES, Dept. of Biosciences, University of Oslo, NO-0316 Oslo, Norway 
  			(\email{n.c.stenseth@mn.uio.no}).}
}
\usepackage{amsopn}
\DeclareMathOperator{\diag}{diag}

\ifpdf
\hypersetup{
  pdftitle={Modeling the process of speciation using a multi-scale framework including a posteriori error estimates},
  pdfauthor={M. K. Brun, E. Ahmed, J. M. Nordbotten, and N. C. Stenseth}
}
\fi

\begin{document}
\nolinenumbers
\maketitle

\begin{abstract}
This paper concerns the modeling and numerical simulation of the process of speciation. In particular, given conditions for which one or more speciation events within an ecosystem occur, our aim is to develop the necessary modeling and simulation tools. Care is also taken to establish a solid mathematical foundation on which our modeling framework is built. This is the subject of the first half of the paper. The second half is devoted to developing a multi-scale framework for eco-evolutionary modeling, where the relevant scales are that of species and individual/population, respectively. The species level model we employ can be considered as an extension of the classical Lotka-Volterra model, where in addition to the species abundance, the model also governs the evolution of the species mean traits and species trait covariances, and in this sense generalizes the purely ecological Lotka-Volterra model to an eco-evolutionary model. Although the model thus allows for evolving species, it does not (by construction) allow for the branching of species, i.e., speciation events. The reason for this is related to that of separate scales; the unit of species is too coarse to capture the fine-scale dynamics of a speciation event. Instead, the branching species should be regarded as a population of individuals moving along a selection of trait axes (i.e., trait-space). For this, we employ a trait-specific population density model governing the dynamics of the population density as a function of evolutionary traits. At this scale there is no a priori definition of species, but both species and speciation may be defined a posteriori as e.g., local maxima and saddle points of the population density, respectively.  
Hence, a system of interacting species can be described at the species level, while for branching species a population level description is necessary. 
Our multi-scale framework thus consists of coupling the species and population level models where speciation events are detected in advance and then resolved at the population scale until the branchin is complete.
Moreover, since the population level model is formulated as a PDE, we first establish the well-posedness in the time-discrete setting, and then derive the a posteriori error estimates which provides a fully computable upper bound on an energy-type error, including also for the case of general smooth distributions (which will be useful for the detection of speciation events). 
Several numerical tests validate our framework in practice.
\end{abstract}



\section{Introduction}
\label{sec:intro}
Mathematical models have a long history in ecology and evolutionary biology, with the most famous example being the Lotka-Volterra model~\cite{leigh1968ecological}, which describes the interaction of two species (commonly referred to as predator and prey) at a timescale where individual traits remain constant, i.e., no evolution. Going beyond strictly ecological models, the interaction between ecology and evolution has been studied extensively in the so-called ``adaptive dynamics" literature (see e.g.,~\cite{dieckmann1997can,dieckmann2007adaptive,gavrilets2004fitness,geritz2000adaptive}, and the references therein). In particular, the concept known as ``adaptive speciation", i.e., the idea that a series of small adaptive changes in the traits of individuals over long enough time leads to a diversity of species, has received considerable attention~\cite{dieckmann1999origin,dieckmann2004adaptive,doebeli1996quantitative,doebeli2000evolutionary,doebeli2003speciation,rosenzweig1978competitive}. In the current work, we are studying the process of adaptive speciation at the level of model formulation where the \emph{species} is regarded as the fundamental unit of the eco-evolutionary system, and where the species dynamics is emergent from the dynamics at the population/individual level.

The process of \emph{speciation} is a biological phenomenon in which a population within one species gradually evolves into two (or more) distinct species (for classical literature on speciation, see e.g.,~\cite{coyne2004speciation,gavrilets2004fitness,nosil2012ecological,price2008speciation}). As such, speciation is an emergent property of natural selection acting on a population of individual organisms. If the individual is regarded as the fundamental unit of the evolutionary process, the point at which the gradual evolutionary changes has accumulated sufficiently to produce a new species becomes a matter of definition (e.g., reproductive isolation). On the other hand, as species is regarded as the fundamental unit, it becomes important to separate evolution into two categories; \emph{cladogenesis}, which is the splitting of a parent species into new distinct child species, and \emph{anagenesis}, which is the gradual evolution of a species that continues to exists. Modeling the process of speciation using mathematical models therefore presents a different set of challenges whether one takes the individual or the species as the fundamental unit of the eco-evolutionary system. 

The difficulty in modeling the process of speciation using any species interaction-type model is inherent in the model itself, i.e., it is assumed a-priori that distinct species can be identified at all times within the population. Since a speciation event will necessarily imply some ambiguity in the identification of distinct species (at least for a period of time), a breakdown in the underlying assumptions of the model is therefore unavoidable. In mathematical terms; the \emph{coarse-scale} species level model is by construction unable to capture the \emph{fine-scale} population dynamics of a speciation event. The mathematical challenge in modeling speciation as an emergent property is therefore related to the connection of these two scales. Relevant studies concerning this difficulty has covered e.g., evolutionary branching driven by stochastic mutations~\cite{vukics2003speciation}, derivation of an eco-evolutionary model at the species level but not including the variability in species abundance~\cite{debarre2014multidimensional}, and the derivation of a model describing the interaction of different morphs within the same species~\cite{sasaki2011oligomorphic}.  

A fundamental observation of real biological systems is that most of the time individuals are clearly grouped into distinct species. Thus, it is natural to take the species as the fundamental unit of any eco-evolutionary model. However, another observation is that over evolutionary time speciation events do occur. Hence, from the time there is only the parent species to the time when increasingly diverging traits among the constituent individuals has resulted in new child species, the natural unit to consider is instead the individual (or population). These observations form the basis of our developments; most of the time we assume the species level formulation is the correct description of a given biological system (i.e., distinct species evolve as species), and only for the (short) intermediate time interval between the existence of parent and child species, we temporarily discard the species-centric view and instead regard the relevant population as trait specific distributions. 

The species-interaction model we take as the starting point is formulated as a system of ordinary differential equations (ODE) describing the dynamics of a biological system in which any finite number of species interact and evolve adaptively over evolutionary time (referred to in the sequel as the \emph{species level}, SLM, or \emph{macro-scale} model). This model extends the Lotka-Volterra system by representing each species not only by its abundance, but also by its mean traits coordinate and trait covariance matrix. It originates in a recent work by two of the authors~\cite{nordbotten2016asymmetric}, wherein the model equations were derived by applying a moment closure averaging technique on a deterministic trait-specific population density model (referred to in the sequel as the \emph{population level}, PLM, or \emph{micro-scale} model), under the assumptions that (A1) distinct species can be identified at all times, and (A2) that the species distributions are of a known statistical quality (in this case, the normal distribution). See also related derivations in e.g.,~\cite{debarre2014multidimensional,sasaki2011oligomorphic}.  

The population level model is formulated as a partial differential equation (PDE), and governs the abundance density as a function of evolutionary traits. Moreover, the authors showed in~\cite{nordbotten2016asymmetric} that the species and population level models are consistent, i.e., they describe the same system dynamics as long as the system behavior is such that both models are valid. Still, these descriptions are fundamentally different in that they operate on separate scales; the species level model operates with the coarser unit of distinct species, while the population level model operates with the finer unit of the individual (i.e., the category of species is not imposed upon the individuals who make up the population). Thus, the species level model described above approximates the system dynamics governed by the population level model, under the conditions that the population can be grouped into distinct species, and that each species' distribution conforms to a normal distribution. While this requirement that trait distributions of a species are essentially normally distributed may seem restrictive, we note that the shape of a distribution is not an intrinsic quality of the observed system, but is equally dependent on the scale of measurement. This assumption can thus be considered not as much as a limitation on the biological system, as a constraint on the scale of measurement. Note also that due to the deterministic nature of the population level model, this property is also inherited at the species level.

The first part of our developments concerns monitoring the species in the system and estimating the deviation from the true population density function. We do this as follows: Since each species is described by an abundance-trait-covariance tuple, we map these to time dependent distributions. Then, using derived a posteriori error estimates of the population level model, we track the modeling error of each species' associated distribution function. The second part concerns the case of residual blow-up, which implies a model breakdown at the species level, here associated with a speciation event. Concerning the actual speciation event we compare two different approaches with regards to splitting the parent species into new child species;
(1) \emph{Heuristic approach}: When a residual blow-up is detected, we proceed to split the relevant reconstructed distribution function along the trait directions orthogonal to the direction of divergence. Each of these sub-distributions are then mapped back to abundance-trait-variance tuples, thus yielding the new child species to be incorporated at the species level. (2) \emph{Multi-scale approach}: When a residual blow-up is detected, we delegate the relevant species' reconstructed density function to the population level model to be solved in an appropriate local region of trait space while coupled to the species level model which governs the interaction of the remaining species in the system. Simultaneously, we map the local abundance density function to  abundance-trait-variance tuples and measure the distance between the mean trait coordinates. At such time when the new child species are sufficiently separated in trait space (e.g., by a multiple of the maximum trait standard deviation), we incorporate the new child species at the species level and decouple the multi-scale model. We assess the accuracy of each approach by comparing the species parameters (pre- and post-speciation) to the corresponding statistical moments of the reference (global) PLM solution. 

Based on our results, we find that the multi-scale approach gives the better approximation of the reference PLM solution post-speciation, and thus conclude that a multi-scale approach is indeed necessary for robust eco-evolutionary modeling at the level of species interaction, capable of handling speciation events. Although we base our developments on the specific models described above, our methodology is more general; it is applicable on any species interaction model which, including the species abundance, also incorporates the evolution of species mean traits and species trait covariance in some way (thus allowing for the reconstruction of population density distributions), whether the model is based on deterministic or stochastic dynamics. Figure \ref{fig:jan} below shows a schematic representation of the micro/macro-scale coupling strategy for speciation events.

\begin{figure}[h]
\centering
  \includegraphics[width=.6\linewidth]{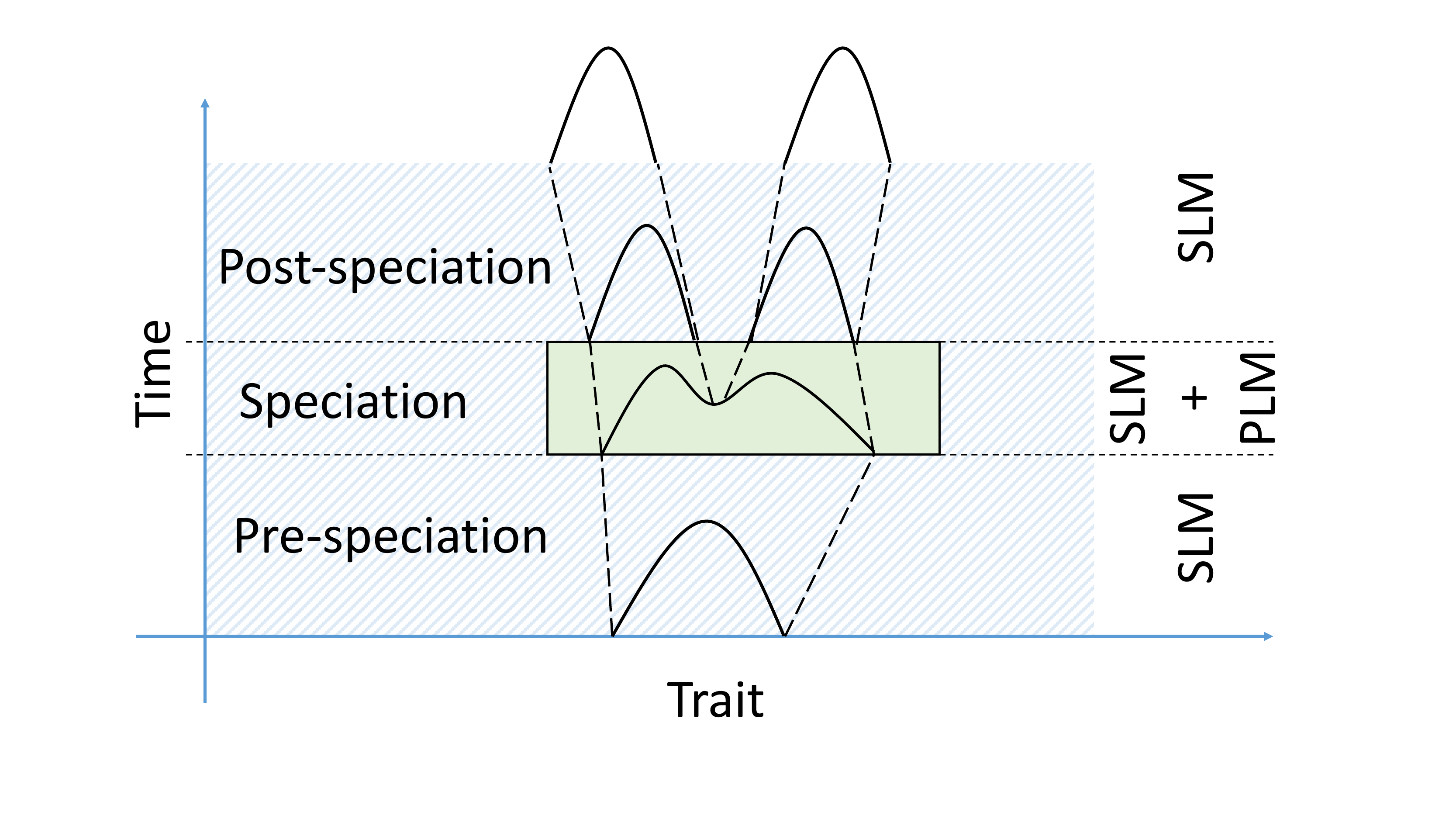}
 \caption{Coupling of SLM and PLM during speciation events. Outside the green area the macro-scale (SLM) model is solved (here, the curve represents the reconstructed distribution function associated with each species' abundance-trait-covariance tuple), the green area indicates the micro-scale (PLM) model is solved (here, the curve represents the abundance density, or, solution function of the PLM).}
\label{fig:jan}
\end{figure}

The article is organized as follows: In Section \ref{sec:models} we present the relevant population and species level models. In Section \ref{sec:timedisc} we introduce the time discrete population level model and proceed to analyze this. In Section \ref{sec:eestimates} we derive the a posteriori error bound of the population level model. In Section \ref{sec:speciation} we present in detail the methodology for species splitting and micro-macro scale coupling. In Section \ref{sec:detection} we refine the derived error bound for the case of individual species.  In Section \ref{sec:ex} we provide numerical examples where we test our multi-scale framework in detail. Finally, in Section \ref{sec:conclusions} we provide some concluding remarks.
\section{Models}
\label{sec:models}
In this section we present both the (micro-scale) population level model (PLM) and (macro-scale) species level model (SLM) as described in~\cite{nordbotten2020dynamics}. We also briefly discuss the consistency of these model formulations.
\subsection{The trait-specific population level model}
The PLM describes the abundance of a population in terms of the population density as a function of evolutionary traits over evolutionary time. The model relates the rate of change in population density to birth/death rate of individuals, interaction between individuals with different evolutionary traits (cooperation/competition), and interaction between individuals with equal evolutionary traits (self-limitation). The evolution of the population, i.e., change in traits from one generation to the next, is incorporated as a diffusion process.  

Given an open, bounded domain $\Omega \subset \real^d$ in trait space (i.e., the number of evolutionary traits considered is equal to $d \geq 1$), and final time $T>0$, let $\Omega_T := \Omega \times (0,T)$ be the space-time domain, and let $n: \Omega_T \rightarrow \real$ be the trait-specific population density, or, number of individuals per measure on $\real^d$. The domain $\Omega$ can then be regarded as all evolutionary traits attainable by the relevant population over evolutionary time, and the total number of individuals present at any time $t \in (0,T)$ is then given by integration over all attainable traits, i.e., $\int_\Omega n(x,t) \dx$, where we require non-negativity of the population density, i.e., 
\begin{equation}
n \geq 0, \textn{ for all } (x,t) \in \Omega_T.
\end{equation}
Furthermore, let $f: \Omega_T \rightarrow \real$ be the source term (net rate of migration/immigration). The model then reads as the following nonlinear and non-local equation:
\begin{equation}
\pt n - rn + \Phi(n)n -\nabla \cdot(g\nabla n) = f. \label{PDE} 
\end{equation}
The ecological processes consists of the growth term, $-rn$, where $r : \Omega_T \rightarrow \real$ is the inherent per-capita growth-rate, and the interaction term, 
\begin{equation}\label{def_interaction}
\Phi(n)n := bn^2 - n \int_\Omega \alpha(x,y,t)n(y,t)\dy,
\end{equation}
where $b : \Omega_T \rightarrow \real_{\geq 0}$ is the local individual limitation coefficient, 
and $\alpha : \Omega \times \Omega_T \rightarrow \real$ is the non-local individual interaction coefficient, i.e., $\alpha(x,y,t)$ is the effect of an individual with trait coordinate $y$ on an individual with trait coordinate $x$. The diffusive term represents the evolutionary process, where $g : \Omega_T \rightarrow \real^{d\times d}$ is the intergenerational trait diffusion tensor. Generally speaking, the model equation \eqref{PDE} belongs to the class of reaction-diffusion type parabolic conservation equations (in this case, with a non-local reaction term). We remark that several generalizations of \eqref{PDE} are possible.  In particular, the diffusive term as stated models incremental evolutionary processes: To allow for rare but possibly non-incremental evolutionary processes more general non-local operators (such as e.g., fractional derivatives), should be considered.

Since homogeneous boundary conditions of Dirichlet type is always justified from a biological perspective given a large enough trait domain $\Omega$, we shall only consider this situation for the present purposes. We let initial data be specified by $n(x,0) = n_0(x)$. The PLM then reads as the following initial/boundary value problem:

Find $n:\Omega_T \rightarrow \real$ such that $n \geq 0$ and
\bse
\begin{alignat}{3}
\pt n - rn + \Phi(n)n - \nabla \cdot (g\nabla n) &= f, \quad &\textnormal{ in } &\Omega_T \label{model1}\\
n &= 0, &\textnormal{ on } &\partial \Omega \times (0,T), \label{modelBC}\\
n &= n_0, &\textnormal{ in } &\Omega \times \{0\}. \label{modelIC}
\end{alignat}
\ese
For additional details regarding the above model, see \cite{nordbotten2016asymmetric}. 
\subsection{The species level model}
The SLM describes the temporal evolution of the abundance, mean traits coordinate, and trait covariance matrix of all species in an eco-evolutionary system. Similarly to the PLM, the rate of change of the species abundance is related to the growth/death rate, self-limitation and cooperation/competition, but here among distinct species. Evolution is thus incorporated in the model by the change in mean traits (i.e., location in trait space) and change in trait covariance (i.e., spread in trait space). 

Indexing the species present in the ecosystem by $i = 1,\cdots,s$ (where $s \geq 1$ is the total number of species), we denote the abundance, mean traits coordinate, and trait covariance matrix for species $i$ at time $t \in (0,T)$ by $n_i(t) \in \real$, $x_i(t) \in \real^d$ and $\upsilon_i(t) \in \real^{d\times d}$, respectively, and let $N_i := (n_i,x_i,\upsilon_i)$ be the tuple representing the $i$'th species, and similarly let $N := (N_1,\cdots,N_s)$ represent the full $s$-species ecosystem. Furthermore, let $F_{0,i}(t) \in \real$ be the source of individuals for species $i$, with associated mean and covariance $F_{1,i}(t) \in \real^d$ and $F_{2,i}(t) \in \real^{d\times d}$, respectively. The model then reads as the following system of ODE's (for $i = 1,\cdots,s$):
\bse
\begin{alignat}{2}
\frac{\textnormal{d}n_i}{\textnormal{d}t} &= \vecN_i(N) 
:= R_in_i - B_in_i^2 + n_i \sum_{\substack{j = 1 \\ j\neq i}}^{s}A_{i,j}n_j + F_{0,i}, 
&t \in (0,T),\label{ODE1}\\
\frac{\textnormal{d}x_i}{\textnormal{d}t} &= \vecX_i(N)
:= \upsilon_i \left(\nabla r_i - \frac{n_i}{2\Lambda(\upsilon_i)}\nabla b_i + \sum_{j=1}^{s} n_j \nabla_i \alpha_{i,j} + F_{1,i}\right), \
&t \in (0,T),\label{ODE2}\\
\frac{\textnormal{d}\upsilon_i}{\textnormal{d}t} 
&= \Sigmabo_i(N) := V_{0,i} + V_{1,i} \upsilon_i + \upsilon_iV_{2,i}\upsilon_i + F_{2,i}, 
&t \in (0,T),\label{ODE3}
\end{alignat}
\ese
where $\Lambda(\upsilon_i) := \sqrt{(2\pi)^d}|\upsilon_i|^{\frac{1}{2}}$ and where (for $i,j = 1,\cdots,s$)
\bse
\begin{alignat}{2}
R_i &:= r_i + \frac{1}{2} \nabla \nabla r_i:\upsilon_i, \\
B_i &:= \frac{1}{\Lambda(\upsilon_i)} \left(b_i + \frac{1}{4} \dnabla b_i : \upsilon_i \right) 
- \left(\alpha_{i,i} + \dnablaii \alpha_{i,i} : \upsilon_i\right), \\
A_{i,j} &:= \alpha_{i,j} + \frac{1}{2} \dnablaii \alpha_{i,j} : \upsilon_i
+ \frac{1}{2} \dnablajj \alpha_{i,j} : \upsilon_j, \\
V_{0,i} &:= 2g_i, \\
V_{1,i} &:= \frac{n_i b_i}{2\Lambda(\upsilon_i)}, \\
V_{2,i} &:= \nabla \nabla r_i + \frac{1}{4} \sum_{j=1}^s \nabla_i \nabla_i \alpha_{i,j}n_j 
+ \frac{n_i}{4\Lambda(\upsilon_i)} \left(\frac{1}{2}(\upsilon_i : \nabla \nabla b_i)\upsilon_i^{-1} - \nabla \nabla b_i\right).
\end{alignat}
\ese
Here, we employed the following notational convention: $r_i(t) := r(x_i(t),t)$, $\nabla r_i(t) := (\nabla r)(x_i(t),t)$ and $\nabla \nabla r_i(t) := (\nabla \nabla r)(x_i(t),t)$, and similarly for the other PLM coefficients. Subscripts on differential operators indicate on which argument the operator is acting. For the full $s$-species ecosystem, we abbreviate the above model as 
\begin{equation}
\Ncal(N) := (\vecN_1(N),\vecX_1(N),\Sigmabo_1(N),\cdots,\vecN_s(N),\vecX_s(N),\Sigmabo_s(N)),
\end{equation}
and let initial data be given by
\begin{equation}
N(0) = N_{0} := (N_{1,0},\cdots,N_{s,0}).
\end{equation}
The species level model then reads as the following initial value problem:

Find $N :(0,T) \rightarrow \real^{s(1+d+d^2)}$ such that $n_i \geq 0$, for $1\leq i \leq s$, and 
\bse
\begin{alignat}{2}
\dfrac{\textnormal{d}N}{\textnormal{d}t} &= \Ncal(N), \quad && t\in(0,T), \label{ODEmodel}\\
N &= N_0, &&t = 0. \label{ODEmodelIC}
\end{alignat}
\ese
In practice, the problem \eqref{ODEmodel}--\eqref{ODEmodelIC} can be reduced to $s(1+d+d(d+1)/2)$ equations due to the symmetry of the trait covariance matrix. For additional details including a derivation of the above model, see~\cite{nordbotten2020dynamics}. Note also that when indexing the components of a point $x\in\Omega$ we use subscripts, i.e., $x = (x_1,\cdots,x_d)$, not to be confused with the species mean trait coordinate $x_i$, where we write $x_i(t) = (x_{i,1}(t),\cdots,x_{i,d}(t))$. 
\subsection{Relationship between population and species level models}
If there is a time interval $J := (t_{0},t_1) \subset (0,T)$ where $s \geq 1$ species can unambiguously be identified, there exists a collection of non-overlapping (possibly time dependent), open and bounded regions in trait space, i.e., $\{B_i \subset \Omega, i \leq s,  B_i \cap B_j = \emptyset, i \neq j\}$, where the population density is compactly supported and contains only one local maximum. Thus, if we let $\mu_{B_i}^k(n)$, be the $k$'th order moment of $n$ (for $k \in \{0,1,2\}$, and assuming $n$ is smooth enough for these moments to exist), within the region $B_i$, i.e.,
\bse
\begin{alignat}{1}
\mu_{B_i}^0(n) &:= \int_{B_i} n(x)\dx,\label{m0}\\ 
\mu_{B_i}^1(n) &:= \frac{1}{\mu_{B_i}^0(n)} \int_{B_i} n(x)x\dx,\label{m1}\\
\mu_{B_i}^2(n) &:= \frac{1}{\mu_{B_i}^0(n)} \int_{B_i} n(x)(x-\mu_{B_i}^1)\otimes (x-\mu_{B_i}^1)\dx,\label{m2}
\end{alignat}
\ese
we can initialize the SLM at $t=t_{0}$ with initial data
\begin{equation}
N_{0} = (\mu_{B_{1}}^0,\mu_{B_{1}}^1,\mu_{B_{1}}^2,\cdots,\mu_{B_{s}}^0,\mu_{B_{s}}^1,\mu_{B_{s}}^2)|_{t=t_0}.
\end{equation}
Solving \eqref{ODEmodel}--\eqref{ODEmodelIC} then amounts to approximating \eqref{m0}--\eqref{m2} for $t \in J$.
\section{Time discrete PLM}\label{sec:timedisc}
In this section we introduce the time discrete PLM, which will serve as the basis for the a posteriori error bounds: 

For $k = 0,\cdots,M$, we let $t_k$ denote the discrete times, such that $0 = t_0 < t_1 < \cdots < t_M := T$, and where $t_k = t_{k-1} + \tau_k = \sum_{j = 1}^k \tau_k$, for some set of time increments $\{\tau_k\}_{k=1}^M$ (in the following, we use superscripts to indicate dependency on the discrete times, e.g., $u^k := u(t_k)$). The time-discrete version of eq.~\ref{PDE} then reads as (for $k\geq 1$)
\begin{equation}
n^k - \tau_k r^k n^k + \tau_k \Phi^k(n^k)n^k - \tau_k \nabla \cdot (g^k \nabla n^k) = \tau_k f + n^{k-1},\label{timediscPDE}
\end{equation}
where we have discretized in time using a standard implicit first order method (i.e., the backward Euler method). While we do not prove convergence of the above discrete scheme to the original continuous problem, the local truncation error for the backward Euler method is known to be of second order, and one would expect first-order convergence of the time-discrete problem to the continuous problem for sufficiently smooth problems. 

We remark that the time-discrete problem has significant biological relevance by itself, as many processes are naturally modeled on a discrete time-scale due to the presence of strong temporal cycles (e.g., day, year, and generation).

In the following, we first establish the existence of a solution to \eqref{timediscPDE} in the weak sense for every discrete time, then we derive an identity for the residual which also provides uniqueness of the time-discrete solution.
\subsection{Notation}
We employ standard notation for function spaces. For a function space $V$ we denote by $\norm{\cdot}_V$ its energy norm, by $V^\ast$ its dual space, and by $\langle \cdot, \cdot \rangle_{V^\ast,V}$ the dual pairing. In particular, for $1 \leq p \leq \infty$, and a domain $D \subset \Omega$ we denote by $L^p(D)$ the Lebesgue-spaces of integrable functions defined on $D$, with associated norm $\norm{\cdot}_{p,D}$ (domain subscript omitted if this is clear from the context). Furthermore, let $L^p_{+}(D) := \{u\in L^p(D) : u \geq 0 \textn{ a.e.}\}$. For the case $p=2$, we let $\norm{u} := \norm{u}_{2} = (u,u)^{\frac{1}{2}}$, where $(\cdot,\cdot)$ is the standard $L^2$ inner product. Moreover, we denote by $\Hdiv$ the space of functions in $[L^2(D)]^d$ admitting a weak divergence, and by $H^1(D)$ the space of functions in $L^2(D)$ admitting weak gradients, with its vanishing trace subspace denoted by $H_0^1(D)$, where $H^{-1}(D) := H_0^1(D)^\ast$. Finally, denote by $H_{0,+}^1(D) := H_0^1(D) \cap L^2_{+}(D)$.

\subsection{Preliminaries}
Regarding the coefficients and source term of the PLM we introduce the following assumption.
\begin{assum}[Coefficients and source term]
\label{assum:coeff}
assume that $r,b,\alpha,g$ and $f$ are defined for all discrete times $t_k$ such that the following holds for all for $1 \leq k \leq M$:
\bse
\begin{enumerate}
\item
$r^k,\, b^k \in L^\infty(\Omega)$, and $\alpha^k \in L^\infty(\Omega \times \Omega)$, such that 
\begin{equation}\label{Rassump}
|(r^ku,u)| \leq R^k\norm{u}^2, \quad \forall u \in L^2(\Omega),
\end{equation}
for some set of positive constants $\{R_k\}_{k=1}^M$.
\item
The time increments $\tau_k$ are chosen such that 
\begin{equation}\label{Rtauassump}
\gamma_k := 1 - \tau_kR_k > 0.
\end{equation}
\item 
$g^k \in [L^\infty(\Omega)]^{d\times d}$, symmetric and uniformly positive definite, and such that 
\begin{equation}\label{Gassump}
 0< G_k|\zeta|^{2}\leq \zeta^{\textn{T}}g^k(x)\zeta, \quad \forall \zeta \in \real^d \setminus \{0\},
\end{equation}
for some set of positive constants $\{G_k\}_{k=1}^M$.
\item $n_0 \in H_{0,+}^1(\Omega)$.
\item $f^k \in H^{-1}(\Omega)$.
\end{enumerate}
\ese
\end{assum}
Due to  \eqref{Rtauassump} and \eqref{Gassump}, we now equip the space $H_{0}^1(\Omega)$ with the following equivalent inner product 
\begin{equation}\label{Xinner}
(u,v)_k := ((1-\tau_kr^k)u,v) + \tau_k(g^k\nabla u, \nabla v), 
\end{equation}
and associated energy norm
\begin{equation}\label{Xnorm}
\vertiii{v}^2_k := (v,v)_k,
\end{equation}
and define the Hilbert space $X_k := (H_{0,+}^1(\Omega),(\cdot,\cdot)_k)$, and corresponding dual space $X_k^\ast$ with respect to $L^2(\Omega)$. For $\xi \in X_k^\ast$, we denote the dual norm by
\begin{equation}\label{defdualnorm}
\vertiii{\xi}_{k^\ast} := 
\sup_{\underset{\vertiii{v}_k=1}{v\in X_k,}}\langle  \xi, v \rangle_k,
\end{equation}
where $\langle  \cdot, \cdot \rangle_k := \langle  \cdot, \cdot \rangle_{X_k^\ast,X_k}$ is the dual pairing. Furthermore, let $\Lcal(X_k,X_k^\ast)$ be the space of all linear maps from $X_k$ into $X_k^\ast$. 
For $\Xi \in \Lcal(X_k,X_k^\ast)$, we denote the operator norm by
\begin{equation}
\vertiii{\Xi}_{k,k^\ast} := 
\sup_{\underset{\vertiii{v}_k=1}{v\in X_k,}} \vertiii{\Xi(v)}_{k^\ast}.
\end{equation}
Due to Assumption~\ref{assum:coeff} and the Poincar\'e inequality, there exists a constant $c_{k,\Omega} > 0$ (depending on $r^k$, $g^k$ and $\tau_k$ in addition to the domain $\Omega$) such that
\begin{equation}\label{Xkcontembedding}
\norm{v} \leq c_{k,\Omega}\vertiii{v}_k, \quad \forall v\in X_k.
\end{equation}
For $1\leq k \leq M$ we now introduce the following assumption regarding the nonlinearities of the PLM:
\begin{assum}[Nonlinear terms]
\label{assum:analysis}
Assume that $\Phi^k(u) \in \Lcal(X_k,X_k^\ast)$ for any $u\in X_k$, such that
\bse
\begin{enumerate}
\item
\begin{equation}
\langle \Phi^k(u)v,v\rangle_k \geq 0, \qquad \forall v\in X_k.
\end{equation}
\item
For every ball $\Bcal_a \subset X_k$ of finite radius $a > 0$, there exists $L^k > 0$ (depending on $k$ and $a$) such that there holds
\begin{equation}\label{PhiLip}
\vertiii{\Phi^k(u) - \Phi^k(v)}_{k,k^\ast} \leq L^k\norm{u -v}_{1} \quad \forall u,v \in \Bcal_a.
\end{equation}
  \item 
  $\Phi^k(u)$ is a monotone operator, i.e., there holds 
  \begin{equation}\label{monotonicity}
  \langle \Phi^k(u)u - \Phi^k(v)v, u - v \rangle_k \geq 0, \quad \forall u,v \in X_k.
\end{equation}
\end{enumerate}
\ese
\end{assum}
With this, we introduce now the weak time-discrete formulation of the problem \eqref{model1}--\eqref{modelIC}: For $k \geq 1$ and given $n^{k-1} \in X_k$ find $n^k \in X_k$ such that the following integral equality holds 
\bse\label{weak_form}
\begin{align}
(n^k,v)_k + \tau_k\langle\Phi^k(n^k)n^k,v\rangle_k = \tau_k\langle f^k,v\rangle_k + (n^{k-1},v), \qquad \forall v \in X_k, \label{weak1}
\end{align}
and such that initial condition \eqref{modelIC} is satisfied in the weak sense, i.e.,  
\begin{equation}\label{weakIC}
(n^0,v) = (n_0,v), \qquad \forall v \in X_k.
\end{equation}
\ese
\begin{remark}[Assumptions]
Assumption \ref{assum:coeff} is, to the best of the authors' understanding, natural from a biological point of view. On the other hand, Assumption \ref{assum:analysis} is introduced for the analysis, and may be more restrictive than desired for certain biological systems. We will return to this issue in the numerical examples (cf. Section~\ref{sec:ex}).
\end{remark}
\subsection{Solvability}
In this section we discuss the solvability of the weak time-discrete PLM \eqref{weak_form}. We summarize the result in the following theorem.
\begin{theorem}[Well posedness]\label{thm:wellposed}
For $d \geq 3$ and given Assumptions \ref{assum:coeff} and \ref{assum:analysis}, then for every $1\leq k \leq M$ there exists a unique solution $n^k \in X_k$ to problem \eqref{weak_form} satisfying
\begin{equation}
\vertiii{n^k}_k \leq \tau_k\vertiii{f^k}_{k^\ast} + c_{k,\Omega}\norm{n^{k-1}}.
\end{equation}

\end{theorem}
\begin{proof}
The proof follows by a series of calculations done in the next sections. First, we show existence of a linearized problem in~\ref{WPlin}, which we then use to infer existence to the nonlinear problem in~\ref{WPnonlin}. Then, we derive an identity for the dual norm of the residual in~\ref{subsec:residual}, which we then apply to obtain the uniqueness in~\ref{WPuniqueness}. Our existence proof is based on techniques from~\cite{alekseev2016stability}.
\end{proof}
\subsubsection{Linearization}\label{WPlin}
We establish the well-posedness of a linearized problem. Choose $\hat{\Phi} \in \Lcal(X_k,X_k^\ast)$ such that $\langle \hat{\Phi} u, u \rangle_k \geq 0$ for all $u \in X_k$, and define the bilinear form $a^k : X_k \times X_k \rightarrow \real$ by 
\begin{equation}
a^k(u,v) := (u,v)_k + \tau_k\langle\hat{\Phi}u,v\rangle_k.
\end{equation}
The linearized problem then reads: For $k \geq 1$ and given $n^{k-1} \in X_k$, find $\hat{n}^k \in X_k$ such that 
\begin{equation}\label{linearized}
a^k(\hat{n}^k,v) = \tau_k \langle f^k,v\rangle_k + (n^{k-1},v), \qquad \forall v \in X_k.
\end{equation}
The form $a^k$ is continuous and coercive on $X_k$:
\bse
\begin{align}
|a^k(u,v)| &\leq (1 + \tau_k\vertiii{\hat{\Phi}}_{k,k^\ast})\vertiii{u}_k \vertiii{v}_k,\\
a^k(u,u) &\geq \vertiii{u}_k^2.
\end{align}
\ese
Thus, by the Lax-Milgram Theorem~\cite{evans1998partial} there exists a unique solution $\hat{n}^k \in X_k$ to \eqref{linearized} satisfying
\begin{align}
\nonumber \vertiii{\hat{n}^k}_k &\leq \vertiii{\tau_kf^k + n^{k-1}}_{k^\ast} \\
&\leq \tau_k\vertiii{f^k}_{k^\ast} + c_{k,\Omega}\norm{n^{k-1}}, \label{linsol}
\end{align}
since the action of $\tau_kf^k + n^{k-1} \in X_k^\ast$ on $X_k$ is defined by the right hand side of \eqref{weak1}, and where we used the triangle inequality and \eqref{Xkcontembedding} in the second line.
\subsubsection{Nonlinear problem}\label{WPnonlin}
We establish existence of a solution to the nonlinear problem. For given $w \in X_k$, let $\hat{n}^k_w \in X_k$ be the unique solution to the following linear problem
\begin{equation}\label{wlinearized}
(\hat{n}^k_w,v)_k + \tau_k\langle\Phi^k(w)\hat{n}^k_w,v\rangle_k = \tau_k \langle f^k,v\rangle_k + (n^{k-1},v).
\end{equation}
Define the map $\Fcal^k : X_k \rightarrow X_k$ by $\Fcal^k(w) = \hat{n}^k_w$. Due to the result of the previous section this is well defined. Moreover, let 
\begin{equation}
\Bcal^k_a := \{w \in X_k : \vertiii{w}_k \leq \vertiii{\tau_kf^k + n^{k-1}}_{k^\ast} =: a\}.
\end{equation}
Due to \eqref{linsol}, it follows that $\Fcal^k(\Bcal_a^k) \subset \Bcal_a^k$. 

Next step is to show that $\Fcal^k$ is continuous and compact on $\Bcal^k_a$. Let $\{w_j\}_{n=1}^\infty \subset \Bcal^k_a$ be a sequence. Since $X_k$ is a Hilbert space, and since $\{w_j\}_{n=1}^\infty$ is bounded, the Eberlein-\v{S}mulian Theorem~\cite{cheney2013analysis} applies; there exists a subsequence (denoted the same way) and $w \in X_k$ such that $w_j \rightharpoonup w$ weakly in $X_k$. Moreover, since $d \geq 3$, by the Rellich-Kondrachov Theorem~\cite{evans1998partial}, the embedding $X_k \subset L^1(\Omega)$ is compact. Thus, $w_j \rightarrow w$ strongly in $L^{1}(\Omega)$. Now, let $\hat{n}^k_{j} = \Fcal^k(w_j)$. Then $\hat{n}^k_{j}$ solves 
\begin{equation}\label{wjlinearized}
(\hat{n}^k_{j},v)_k + \tau_k\langle\Phi^k(w_j)\hat{n}^k_{j},v\rangle = \tau_k\langle f^k,v\rangle_k + (n^{k-1},v).
\end{equation}
Furthermore, $\hat{n}^k_w = \Fcal^k(w)$ solves \eqref{wlinearized}. This establishes the continuity of $\Fcal^k$. 

Next step is to show $\hat{n}^k_{j} \rightarrow \hat{n}^k_w$ strongly in $X_k$. To this end, subtract \eqref{wlinearized} from \eqref{wjlinearized} to obtain
\begin{equation}\label{lindiff}
(\hat{n}^k_{j} - \hat{n}^k_w,v)_k + \tau_k\langle\Phi^k(w_j)(\hat{n}^k_j - \hat{n}^k_w),v\rangle_k = -\tau_k\langle(\Phi^k(w_j) - \Phi^k(w))\hat{n}^k_j,v\rangle_k.
\end{equation}
Since we have
\begin{align}
\nonumber |\langle(\Phi^k(w_j) - \Phi^k(w))\hat{n}^k_j,v\rangle_k| &\leq \vertiii{(\Phi^k(w_j) - \Phi^k(w))\hat{n}^k_j}_{k^\ast} \vertiii{v}_k \\
\nonumber &\leq \vertiii{\Phi^k(w_j) - \Phi^k(w)}_{k,k^\ast}\vertiii{\hat{n}^k_j}_k\vertiii{v}_k \\
&\leq L^k\norm{w_j - w}_{1}\vertiii{\tau_kf^k + n^{k-1}}_{k^\ast}\vertiii{v}_k, 
\end{align}
it follows that 
\begin{equation}
|\langle(\Phi^k(w_j) - \Phi^k(w))\hat{n}^k_j,v\rangle_k| \rightarrow 0, \qquad \forall v\in X_k,
\end{equation}
and consequently $\hat{n}^k_{j} \rightarrow \hat{n}^k_w$ in $X_k$ as $j \rightarrow \infty$, since the equation \eqref{lindiff} holds for any $v \in X_k$ and is linear in the argument $\hat{n}^k_j - \hat{n}^k_w$. Thus, $\Fcal^k$ is a continuous and compact operator on $\Bcal^k_a$, and the Schauder Theorem~\cite{cheney2013analysis} applies; there exists $n^k \in \Bcal^k_a$ such that $n^k = \Fcal^k(n^k)$, which is the solution to the nonlinear problem \eqref{weak_form}. 
\subsubsection{Residual identity}
\label{subsec:residual}
Based on the previous existence result,  we introduce now the residual operator $\Rcal^k(\psi) \in X_k^\ast$ for a given $\psi \in X_k$ and for all $v \in X_k$ as
\begin{align}
      \label{Residual_operator}
\langle \Rcal^k(\psi), v \rangle_k := \tau_k\langle f^k,v\rangle_k + (n^{k-1}, v) - (\psi,v)_k - \tau_k\langle\Phi^k(\psi)\psi,v\rangle_k.
\end{align}
We also introduce the semi-metric (cf. the monotonicity condition \eqref{monotonicity})
\begin{align}
\nonumber \Jcal^k(n^k,\psi) &:= \tau_k^2\vertiii{\Phi^k(n^k)n^k - \Phi^k(\psi)\psi }^2_{k^\ast} \\
&\qquad+ 2\tau_k\langle\Phi^k(n^k)n^k - \Phi^k(\psi)\psi, n^k - \psi\rangle_k.
\end{align}
The following Lemma now makes precise the connection between the dual norm of the residual and the resulting error. 
\begin{lemma}[Residual-error identity]\label{lem:edist}
For all $\psi \in X_k$ there holds
\begin{equation}\label{res_error}
\vertiii{\mathcal{R}^k(\psi)}^2_{k^\ast} = \vertiii{n^k - \psi}^2_k + \Jcal^k(n^k,\psi).
\end{equation}
\end{lemma}
\begin{proof}
Choose $\phi \in X_k$ such that 
\begin{equation}\label{dualprob}
(\phi,v)_k=\tau_k \langle \Phi^k(n^k)n^k - \Phi^k(\psi)\psi ,v\rangle_k, \qquad \forall v \in X_k.
\end{equation}
Thus, we have the identity
\begin{equation}
\vertiii{\phi}_k = \tau_k\vertiii{\Phi^k(n^k)n^k - \Phi^k(\psi)\psi}_{k^\ast}.
\end{equation}
We can then write the dual norm of the residual as follows
\begin{align}
\nonumber  \vertiii{\mathcal{R}^k(\psi)}_{k^\ast} = &\sup_{\underset{\vertiii{v}_{k}=1}{v\in X_k,}}   \{(n^k-\psi,v)_k+ \tau_k\langle\Phi^k(n^k)n^k - \Phi^k(\psi)\psi,v\rangle_k \}\\
\nonumber =&\vertiii{\phi + n^k-\psi}_k \\
\nonumber =&\Big( \vertiii{\phi}_k^2 + 2(\phi,n^k-\psi)_k  + \vertiii{n^k-\psi}_k^2\Big)^{\frac{1}{2}} \\
\nonumber =&\Big(\tau_k^2 \vertiii{\Phi^k(n^k)n^k - \Phi^k(\psi)\psi }^2_{k^\ast}\\
&\hspace{2cm}+ 2\tau_k\langle\Phi^k(n^k)n^k - \Phi^k(\psi)\psi, n^k-\psi\rangle_k
+ \vertiii{n^k-\psi}_k^2 \Big)^{\frac{1}{2}}.
\end{align}
The assertion follows.
\end{proof}
\begin{remark}[Semi-metric]
Due to the nonlinear nature of the model equation \eqref{weak_form}, the norm of the residual can not be shown to be equal to some energy norm of the error. Instead, it is a combination of a norm and a semi-metric. The semi-metric satisfies non-negativity, symmetry, and identity of indiscernibles, but not the triangle inequality. See e.g., \cite{woodward2000analysis}, where a nonlinear parabolic equation was analyzed and similar terms appear.
\end{remark}
\subsubsection{Uniqueness}\label{WPuniqueness}
It remains to show the uniqueness of $n^k$. To this end, we suppose that
$\psi$ satisfies~\eqref{weak_form}, which implies that $\vertiii{\mathcal{R}^k(\psi)}_{k^\ast}=0$. Then, since the first term on the right hand side of~\eqref{res_error} is a norm and since the second term is convex due to the monotonicity condition~\eqref{monotonicity}, we obtain that $\psi=n^{k}$ in $X_k$. This section concludes the proof of Theorem~\ref{thm:wellposed}.
\section{A posteriori error bound}\label{sec:eestimates}
In this section we adopt energy-type a posteriori error bounds based on the dual norm of the residual for the weak time-discrete PLM, i.e., the problem \eqref{weak_form} using similar techniques as in~\cite{ern2017guaranteed,ern2010posteriori,papevz2018estimating,smears2020simple}. We then provide a guaranteed and fully computable upper bound on the energy--type error in terms of various error estimators, using suitable flux and density reconstructions.

\subsection{Space approximations}
Let $\Zcal_h$ be a partition of the domain $\Omega$, consisting of rectangular  or simplicial elements, such that $\bar{\Omega} = \cup_{K\in \Zcal_h} K$.  We let $h_K$ denote the diameter of the element $K\in \Zcal_h$, and define $h := \max_{K\in \Zcal_h} h_K$ as the maximum diameter over all elements in $\Zcal_h$. For two elements $K, L \in \Zcal_h$, we require their intersection to be either a common face, edge, vertex, or the empty set. We use subscripts to indicate dependency on the discrete mesh $\Zcal_h$, e.g., $u_h$. Our a posteriori error estimates given next will involve  approximations of 
coefficients and source term appearing in~\eqref{weak_form}. Thus, for any function $\varphi^{k}\in \left\{r^{k},\, b^k, \alpha^k,\, g^{k},\, f^{k}\right\}$~(see Assumption~\ref{assum:coeff}), we denote by $\varphi^{k}_{h}$ its approximation satisfying the following orthogonality
\begin{equation}\label{ortho}
 (\varphi^{k}_{h}-\varphi^{k}, v)_{K}=0,\quad \forall v\in L^{2}(\Omega),
\end{equation}
with the approximation of the interaction $\Phi_h^k$ corresponding to \eqref{def_interaction}. Furthermore, for a convex element $K\in\Zcal_h$, there holds the Poincar\'{e}--Friedrichs inequality~(see e.g., \cite{MR2338400}), i.e., 
\begin{equation}
\norm{v - v_k}_K \leq \frac{h_{K}}{\pi} \norm{\nabla v}_K, \quad \forall v \in H^1(K), \label{ineq}
\end{equation}
where $v_k$ is the mean value of $v$ over $K$.
\subsection{First upper bound}
To be agnostic with regards to the scheme (micro or/and macro levels) we use to approximate the weak solution, we make the following definitions:
\begin{df}[Density reconstruction]\label{densityrec}
We call a density reconstruction any function $s_h^k \in X_k$.
\end{df}
\begin{df}[Equilibrated flux and density reconstructions]\label{fluxequi}
We call an equilibrated flux reconstruction any function $\sigmabo_h^k : \Omega \rightarrow \real^d$  which satisfies:
\bse
\begin{alignat}{2}
\label{ass:hdiv_conform} &\sigmabo_h^k\in \Hdiv,&\\
\nonumber &(\tau_k\nabla \cdot \sigmabo^{k}_{h},1)_K &\\
\label{ass:equilibration}  & \ =(\tau_k f^k_h + s_h^{k-1} - (1 - \tau_kr^k_h)s_h^k - \tau_k \Phi^k_h(s_h^k)s_h^k,1)_K, 
\ &1\leq k\leq M, \forall K \in \Zcal_h.
\end{alignat}
\ese
\end{df}
For all $K\in\Zcal_h$,  define the local residual, 
flux, \textit{r}-data, $\Phi$-data, \textit{g}-data, and \textit{f}-data oscillation estimators:
 \bse\label{def:loc_estimators1}
\begin{alignat}{2}
\etaRK^{k} &:= \tilde{\omega}^{k}_{K}\norm{\tau_kf_{h}^k + n^{k-1} - (1 - \tau_kr_h^k)s_h^k -\tau_k\Phi_{h}^k(s_h^k)s_h^k - \tau_k\nabla \cdot \sigmabo^{k}_{h}}_K \label{def:loc_estimators3}\\
\etaDF^k &:= \tau_kG_k^{-\frac{1}{2}}\norm{g_{h}^k \nabla s_h^k + \sigmabo^k_{h}}_K,\label{def:loc_estimators2}\\
\etarh^k &:= \tau_k\tilde{\omega}^{k}_{K}\norm{(r^k - r_h^k)s_h^k}_K,\\
\etaphih^k &:= \tau_k\vertiii{(\Phi^k(s_h^k) - \Phi_{h}^k(s_h^k))s_h^k}_{k^\ast}, \\
\etagh^k &:= \tau_kG_k^{-\frac{1}{2}}\norm{(g^k - g_{h}^k) \nabla s_h^k}_K, \\
\etaosch^{k} &:=  \tau_k\vertiii{f^k -f_h^k}_{k^\ast}.
&  
\end{alignat}
\ese
where we introduced also the following weight
\begin{alignat}{2}
  \tilde{\omega}^{k}_{K}:=\min(\gamma_k^{-\frac{1}{2}}, G_k^{-\frac{1}{2}}\frac{h_{K}}{\pi}).
\end{alignat}
 We  now give the first (abstract)  a posteriori error 
 bound of  the problem~\eqref{weak_form}.
 \begin{theorem}[Energy-error bound]
For $1\leq k \leq M$, let $n^k$ be the (unknown)  exact solution to the weak-time discrete problem~\eqref{weak_form}, $s_h^{k}$ the reconstructed density of Definition~\ref{densityrec}, and $\sigmabo_h^k$ is the reconstructed flux of Definition~\ref{fluxequi}. Then,  the following estimate holds true
\begin{alignat}{2}
&\nonumber\vertiii{n^k - s_h^k}^2_k + \Jcal^k(n^k,s_h^k)& \\\label{first_estimate}
&\qquad \quad\leq \left( \left\{\sum_{K\in \Zcal_h}\left[\etaRK^{k}  + \etarh^k +\etaDF^{k} + \etagh^k \right]^2 \right\}^{\frac{1}{2}} + \etaphih^k + \etaosch^{k}\right)^2.
\end{alignat}
\end{theorem}
\begin{proof}
We  let $\psi = s_h^k$ in~\eqref{Residual_operator}, then we subtract the terms \\ 
$\tau_k(\sigmabo_{h}^k,\nabla v)+\tau_k(\nabla\cdot\sigmabo^k_{h}, v)=0$, followed by adding and subtracting the  terms \\
$\tau_k(r^k_{h} s_{h}^k,v)$, $\tau_k(\Phi_{h}^k(s_h^k)s_h^k,v)$, $\tau_k(g_{h}^k \nabla s_h^k, \nabla v)$, and $\tau_k(f_{h}^k,v)$, which leads to 
\begin{alignat}{2}
\nonumber \langle \Rcal^k(s^k_{h}), v \rangle_k&=(\tau_k f_{h}^k + n^{k-1} - (1 - \tau_kr_h^k)s_h^k -\tau_k\Phi_{h}^k(s^k_{h})s_{h}^k - \tau_k \nabla \cdot \sigmabo_{h}^k, v). \\
\nonumber&\qquad -\tau_k(g_{h}^k \nabla s_h^k + \sigmabo_{h}^k, \nabla v)
+ \tau_k((r^k - r^k_h)s_h^k,v) \\
\nonumber&\qquad + \tau_k\{(\Phi^k_k(s_h^k)s_h^k,v)  -\langle \Phi_{h}^k(s_h^k)s_h^k,v\rangle_k\} \\
\nonumber &\qquad + \tau_k ((g_h^k - g^k) \nabla s_h^k, \nabla v)
+\tau_k \{\langle f^k,v\rangle_k - (f_{h}^k,v)\} \\
&=: \sum_{i=1}^{6} T_i^k, 
\label{resid_op_detailed}
\end{alignat}
for all $v \in X_k$. Recalling the local estimators~\eqref{def:loc_estimators1}, we  estimate each of the terms $T_1^k$--$T_{6}^k$ on the right hand side of~\eqref{resid_op_detailed}, i.e., for the first term, we can use \eqref{Rtauassump} to get  
\begin{align}\label{res_estim1}
\nonumber T_{1}^k \leq \sum_{K\in \Zcal_h}\gamma^{-\frac{1}{2}}_k\norm{\tau_k f_{h}^k + n^{k-1} - (1 - \tau_kr_h^k)s_h^k -&\tau_k\Phi_{h}^k(s^k_{h})s_{h}^k - \tau_k \nabla \cdot \sigmabo_{h}^k}_K \\
&\cdot \norm{(1 - \tau_kr^k)^{\frac{1}{2}}v}_K.
\end{align}
Also, we can use the equilibration property~\eqref{ass:equilibration} together with the orthogonality \eqref{ortho} to obtain
\begin{alignat}{2}
\nonumber T_{1}^k &= \sum_{K\in \Zcal_h}(\tau_kf_{h}^k + n^{k-1} - (1 - \tau_kr_h^k)s_h^k -\tau_k\Phi_{h}^k(s^k_{h})s_{h}^k - \tau_k\nabla \cdot \sigmabo_{h}^k, v)_{K},\\
\nonumber &= \sum_{K\in \Zcal_h}(\tau_kf_{h}^k + n^{k-1} - (1 - \tau_kr_h^k)s_h^k -\tau_k\Phi_{h}^k(s^k_{h})s_{h}^k \\
&\hspace{6cm}- \tau_k\nabla \cdot \sigmabo_{h}^k, v- v_{K})_{K}.
\end{alignat}
By applying the Poincar\'{e}--Friedrichs inequality \eqref{ineq} we get
\begin{align}
\nonumber T_{1}^k \leq G_k^{-\frac{1}{2}}\frac{h_{K}}{\pi}\sum_{K\in \Zcal_h}\norm{\tau_k f_{h}^k + n^{k-1} - (1 - \tau_kr_h^k)s_h^k -&\tau_k\Phi_{h}^k(s^k_{h})s_{h}^k - \tau_k \nabla \cdot \sigmabo_{h}^k}_K\\
\label{res_estim2} &\cdot\norm{(g^k)^{\frac{1}{2}} \nabla v}_K.
\end{align}
Combining~\eqref{res_estim1} and~\eqref{res_estim2}, we obtain
\begin{equation}
T_{1}^k \leq \sum_{K\in \Zcal_h}\etaRK^k  \vertiii{v}_{K, k}.
\end{equation}
The same procedure applied to the term $T^{k}_{3}$ yields
\begin{equation}
 T_{3}^k \leq \sum_{K\in \Zcal_h}\etarh^k \vertiii{v}_{K,k}.
\end{equation}
On the terms $T_2$ and $T_5$ we use \eqref{Gassump} to obtain
\bse
\begin{align}
T_{2}^k &\leq \sum_{K\in \Zcal_h}\etaDF^k\norm{(g^k)^{\frac{1}{2}} \nabla v}_{K},\\
T_{5}^k &\leq \sum_{K\in \Zcal_h}\etagh^k\norm{(g^k)^{\frac{1}{2}} \nabla v}_{K}.
\end{align}
\ese
Finally, on $T_4$ and $T_6$ we use the definition of the dual norm \eqref{defdualnorm} to obtain
\bse
\begin{align}
T_{4}^k &\leq \etaphih^k\vertiii{v}_{k},\\
T_{6}^k &\leq \etaosch^k\vertiii{v}_{k}.
\end{align}
\ese
Combining the above bounds leads to 
\begin{align}
\nonumber &\langle \Rcal^k(s_h^k), v \rangle_k \\
\nonumber &\quad \leq \sum_{K \in \Zcal_h}\left\{\left[\etaRK^{k}  + \etarh^k +\etaDF^{k} + \etagh^k \right]\vertiii{v}_{K,k}\right\} + (\etaosch^{k} + \etaphih^k) \vertiii{v}_k \\
&\quad\leq \left( \left\{\sum_{K\in \Zcal_h}\left[\etaRK^{k}  + \etarh^k +\etaDF^{k} + \etagh^k \right]^2 \right\}^{\frac{1}{2}} + \etaosch^{k} + \etaphih^k\right)\vertiii{v}_k,
\end{align}
since $\vertiii{v}_{k}^2= \sum_{K \in \Zcal_h}\vertiii{v}_{K,k}^2$. By~\eqref{defdualnorm} and~\eqref{res_error} we prove the assertion.
\end{proof}
\section{Speciation}\label{sec:speciation}
In this section we present our strategies for the actual speciation event. The first is based on splitting the support of the relevant reconstructed density function into subregions of trait space, and mapping back to new species abundance-trait-variance tuples by calculating the statistical moments of the reconstructed density in each subregion. The second approach is based on coupling the PLM and SLM in a multi-scale framework for the duration of the speciation event, where the diverging species is delegated to the PLM to be solved locally in trait space.
\subsection{Trait space density regions}
For any integer $\nu = 1,2,\cdots$, and time step $k \geq 1$, we associate with species $i$ the ($d$-hyper-rectangular) density region, $B_{i,\nu}^k \subseteq \Omega$ centered on $x_i^k$ and with orientation according to the (orthonormal) spectrum of $\upsilon^{k}_i(t)$, and with side lengths equal to $2\nu(\lambda_1^k)^{1/2},\cdots,2\nu(\lambda_d^k)^{1/2}$, where $\{\lambda_j^k, j\leq d\}$ are the eigenvalues of $\upsilon_i^{k}$ (i.e., for any time discrete time $t_k$ the density region $B_{i,\nu}^k$ extends $\nu$ times the $j$'th standard deviation of the trait covariance along the $j$'th trait dimension from the mean traits coordinate $x_i^k$). 
\subsection{Heuristic approach}
If at time $t_k$ there is detected diverging traits within a species $N_i^k$, we split this into two child species, $N_{i,1}^k$ and $N_{i,2}^k$, as follows: The associated density region $B_{i,\nu}^k$ is divided along the directions orthogonal to the largest eigenvector of the trait covariance matrix, $\upsilon_i^k$, into two sub-regions $B_{i,\nu}^{k,1}$ and $B_{i,\nu}^{k,2}$ such that $B_{i,\nu}^{k,1} \cup B_{i,\nu}^{k,2} = B_{i,\nu}^k$, and the new species are initialized by the moments of the reconstructed density function of the parent species in each sub-region, i.e., for $j \in \{1,2\}$, we define
\begin{equation}
N_{i,j}^k := (\mu_{B_{i,\nu}^{k,j}}^0(\Dcal N^k_i),\mu_{B_{i,\nu}^{k,j}}^1(\Dcal N^k_i),\mu_{B_{i,\nu}^{k,j}}^2(\Dcal N^k_i)),
\end{equation}
where $\Dcal : \real^{1+d+d^2} \rightarrow C^\infty(\Omega) \cap X_k$ is the density reconstruction operator (i.e., macro-to-micro scale mapping), such that $\Dcal N^k_i$ is any smooth distribution characterized by a mean and variance, and adhering to the relevant boundary condition (in this case, homogenous Dirichlet). Figure \ref{fig1:bi} below shows the splitting of an example density region in 2D, where $x_i^k = (0,0)$ and $\upsilon_i^k = \diag(\upsilon_{i,11}^k, \upsilon_{i,22}^k)$.
\begin{figure}[h]
\centering
\begin{subfigure}{.45\textwidth}
  \centering
  \includegraphics[width=1.\linewidth]{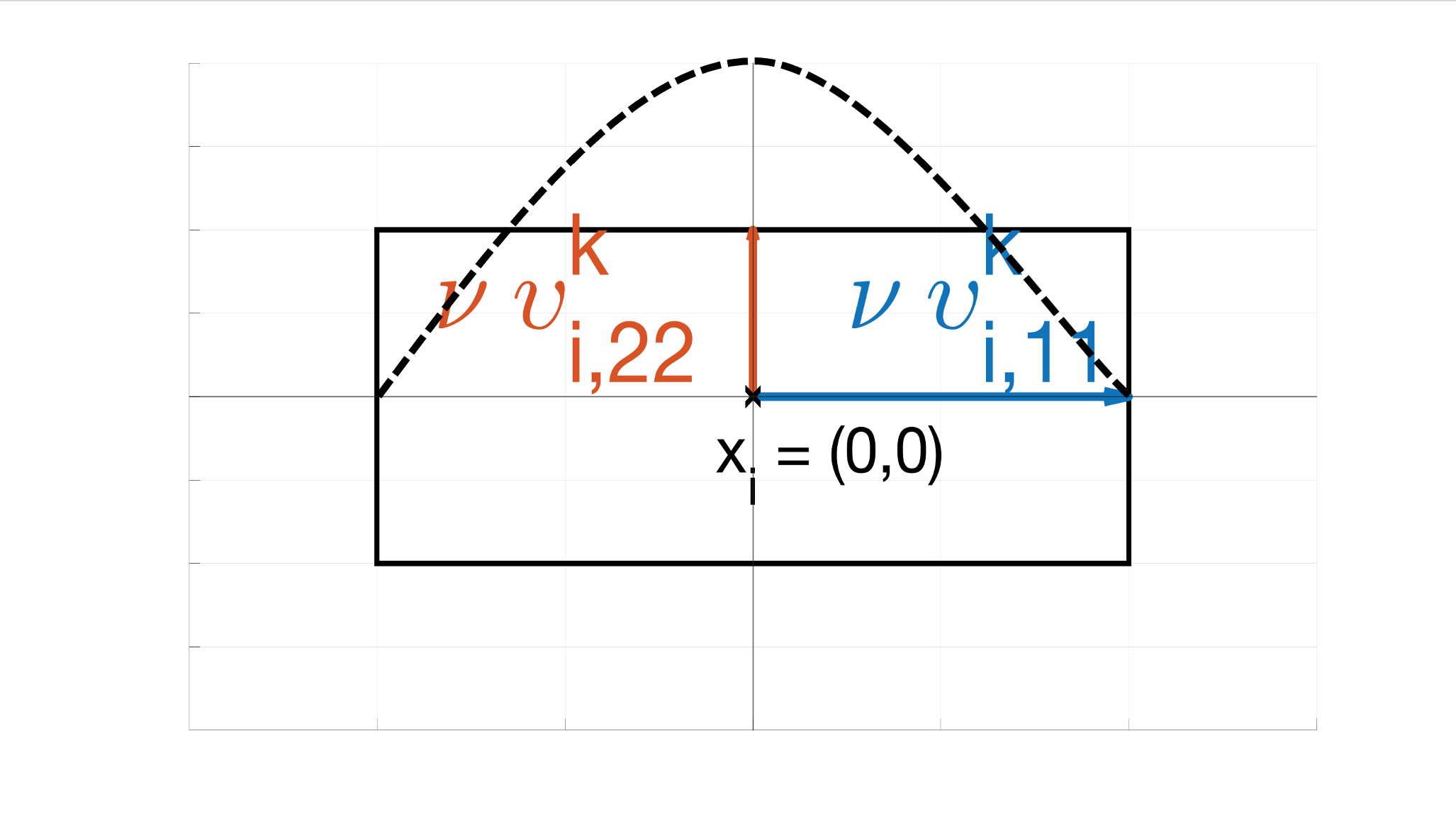}
  \caption{Density region, with spectrum of \\ covariance matrix scaled with $\nu$.}
  \label{subfig11:bi}
\end{subfigure}%
\begin{subfigure}{.45\textwidth}
  \centering
  \includegraphics[width=1.\linewidth]{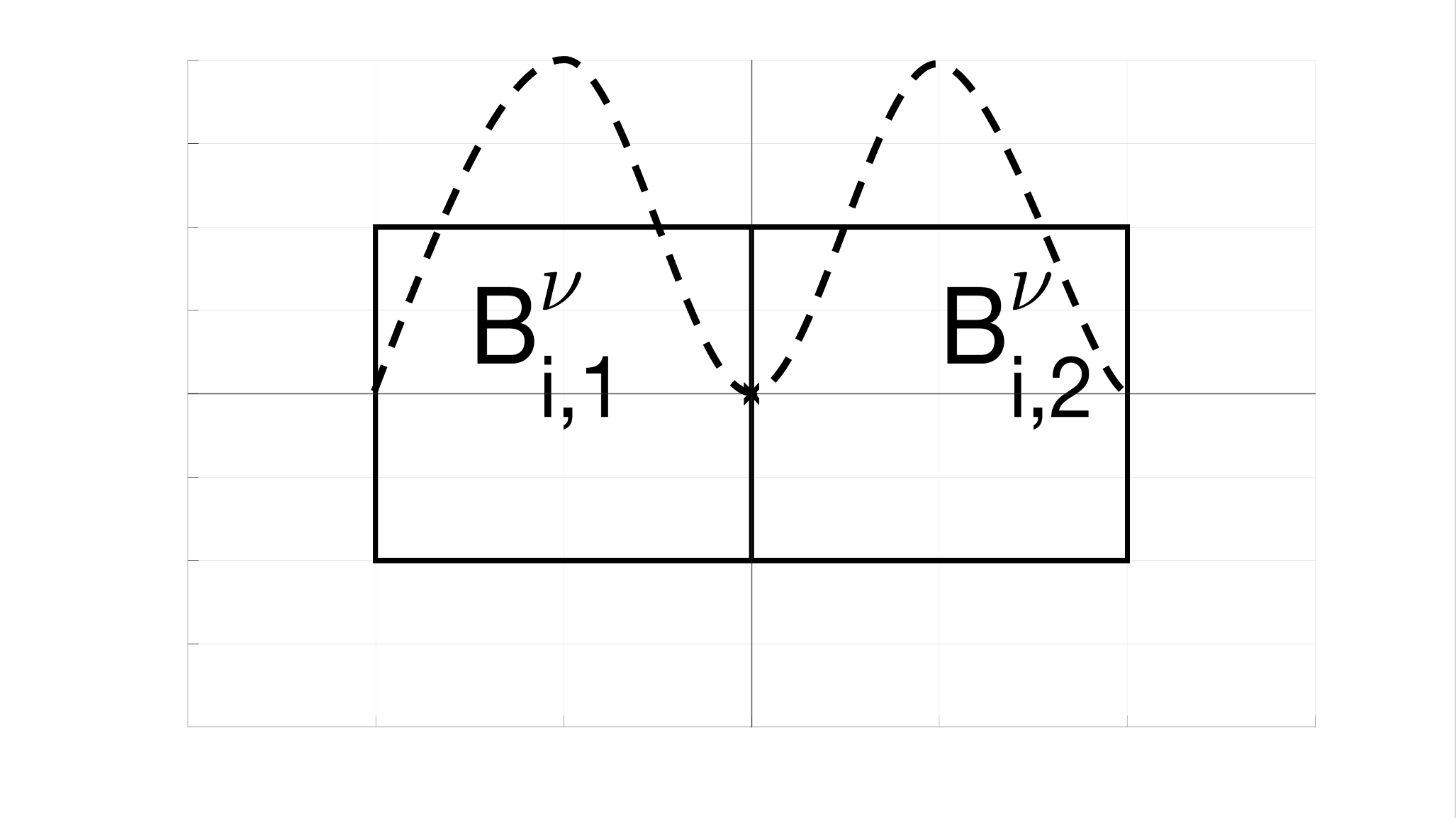}
  \caption{Splitting of density region along direction orthogonal to direction of trait divergence into two subregions.}
  \label{subfig12:bi2}
\end{subfigure}
  \caption{Splitting of example density region in 2D. Dotted line represents species distribution.}
\label{fig1:bi}
\end{figure}
\subsection{Multi-scale approach}\label{sec:micro_macro_coupling}
In this section we couple the PLM, \eqref{model1}--\eqref{modelIC}, and the SLM, \eqref{ODEmodel}--\eqref{ODEmodelIC}, in a multi-scale framework, following~\cite{weinan2011principles}. In particular, if a speciation event is detected for species $i$ at time step $k\geq 1$, we delegate its reconstructed density function to the PLM to be solved locally in the associated density region of trait space, while coupled to the SLM which governs the remaining $s-1$ species in the system. With this, we write the coupled multi-scale model as
\bse
\begin{align}
\nonumber &\textbf{Micro-scale: } \\
&\quad\begin{cases} 
\pt n^{\ell} - rn^\ell + \Phi(n^\ell)n^{\ell} - \nabla \cdot (g \nabla n^\ell) + n^\ell\Cbb^\ell(N^{\ell}) = f^\ell, &\textn{ in } B_{i,\nu}^k, \, \ell > k \label{HMM2} \\
n^k = \Dcal N_i^k, &\textn{ in } B_{i,\nu}^k. 
\end{cases}\\[2ex]
\nonumber &\textbf{Macro-scale: } \\
&\quad\begin{cases}
\dfrac{\textnormal{d}N^\ell}{\textnormal{d}t} = \Ncal(N^\ell) + \Dbb^\ell(n^\ell)N^\ell, &\ell > k \label{HMM1} \\[1ex]
N^k = (N_1^k,\cdots, N_{i-1}^k, 0, N_{i+1}^k, \cdots, N_{s}^k), &
\end{cases}
\end{align}
\ese
where $\pt n^{\ell}$ and $\dfrac{\textnormal{d}N^\ell}{\textnormal{d}t}$ denotes the discrete time derivatives for the PLM and SLM, respectively, and where $\Cbb^\ell : \real^{s(1+d+d^2)} \rightarrow \real$ and $\Dbb^\ell : H_{0,+}^1(B_{i,\nu}^k) \rightarrow \real^{s(1+d+d^2)} $ are the micro/macro couplings, approximating respectively the individual and species interactions. These are defined as follows: 
\bse
\begin{align}
(\Cbb^\ell (N^\ell))(x) &:= \sum_{\substack{j = 1 \\ j \neq i}}^{s} \int_{B_{j,\nu}^\ell} \alpha^\ell(x,y) (\Dcal N_j^\ell) (y)\dy, \\
\nonumber \Dbb^\ell(n^\ell) := (&n_1^\ell A_{1,i}^\ell, \ \nabla_1\alpha_{1,i}^\ell, \ \upsilon_{1}^\ell \nabla_{1}\nabla_{1}\alpha_{1,i}^\ell\upsilon_{1}^\ell/4, \\
\nonumber \cdots,&n_{i-1}A_{i-1,i}, \ \nabla_{i-1}\alpha_{i-1,i}, \ \upsilon_{i-1}^2 \nabla_{i-1}\nabla_{i-1}\alpha_{i-1,i}\upsilon_{i-1}^2/4,\\
\nonumber  0,\,&n_{i+1}^\ell A_{i+1,i}^\ell, \ \nabla_{i+1}\alpha_{i+1,i}^\ell, \ \upsilon_{i+1}^\ell \nabla_{i+1}\nabla_{i+1}\alpha_{i+1,i}^\ell\upsilon_{i+1}^\ell/4,\\
\cdots,&n_{s}^\ell A_{s,i}^\ell, \ \nabla_{s}\alpha_{s,i}^\ell, \ \upsilon_{s}^\ell \nabla_{s}\nabla_{s}\alpha_{s,i}^\ell\upsilon_{s}^\ell/4)|_{\Ccal_{B^k_{i,\nu}}(n^\ell)}\int_{B^k_{i,\nu}} n^\ell(x)\dx.
\end{align}
\ese
During the speciation event, we map the local abundance density function to $m$ `virtual' species, i.e.,
\begin{equation}
\hat{N}^\ell = (\hat{N}_1^\ell,\cdots,\hat{N}^\ell_m) := \Ccal^m_{B_{i,\nu}^k}(n^\ell),
\end{equation}
where $ \Ccal^m_{B_{i,\nu}^k}(n^\ell) : H_0^1(B_{i,\nu}^k) \rightarrow \real^{m(1+d+d^2)}$ is the $m$-species compression operator, defined by
\begin{equation}\label{compr}
\Ccal^m_{B_{i,\nu}^k}(n^\ell) := \argmin_{\substack{n_j, x_j, \upsilon_j, \\ 1\leq i \leq m}} \bnorm{n^\ell - \sum_{j=1}^m \Dcal(n_j,x_j,\upsilon_j) }^2_{B_{i,\nu}^k}.
\end{equation}
The speciation event is then complete when the distances between the mean trait coordinates of the virtual species are larger than some specified tolerance (e.g., some multiple of the largest trait standard deviation). At this point, we initialize the SLM with the new species configuration consisting of $s-1+m$ species, i.e.,
\begin{equation}
(N_1^\ell,\cdots,N_{i-1}^\ell,N_{i+1},\cdots,N_s^\ell) \cup \hat{N}^\ell,
\end{equation}
and decouple the micro and macro scales until the next speciation event is detected. 
\section{Application: distinguishing the modeling error}\label{sec:detection}
In this section we apply the results from Section~\ref{sec:eestimates} to derive an energy-type a posteriori error bound for  the heuristic and multi-scale methods. The error bound  distinguishes in particular the micro-to-macro modeling error, which will be  employed for the purposes of detecting speciation events in the SLM.
\subsection{Abundance density reconstruction and flux equilibration  \\ with smooth sources}\label{sub:reconstruction}
To apply the results of Theorem~\ref{first_estimate},  we reconstruct the species density with smooth distributions, and for the sake of simplicity we only consider  smooth data.  At the time step $k\geq0$, for the species $i\leq s$, we construct its density function by setting 
\begin{equation}\label{construct_pop_density}
s_{h,i}^k := \Dcal N^k_i. 
\end{equation}
Thus, the global ($s$-species) abundance density reconstruction is $H^1$-conforming  in space i.e., 
\begin{equation}\label{def:species}
s_h^k := \sum_{i=1}^s s^k_{h,i}\in C^{\infty}(\Omega)\cap X_k.
\end{equation}
What remains is to construct the equilibrated flux, $\sigmabo_h^k$ satisfying Definition~\ref{fluxequi}. The idea is to reconstruct this equilibrated flux as the sum of a discretization flux $\sigmabo_{h,\textn{disc}}^k$ and the remainder flux $\sigmabo_{h,\textn{rem}}^k$. In this section we assume the source term  $f$ is smooth enough to calculate   the discrete residual for single species, $\mathbf{r}^k_{h,i}$   explicitly,   i.e., for $1\leq i \leq s$, $1\leq k \leq M$, let
\begin{alignat}{2}
\mathbf{r}^k_{h,i}&:= \frac{1}{s}(\tau_k f^{k}_h + s_h^{k-1}) -(1 - \tau_kr_h^k)s^{k}_{h,i} -\tau_k\Phi_h^k(s_h^k)s^k_{h,i}+\tau_k\nabla\cdot (g_h^k\nabla s^k_{h,i}),\label{addlabel1}
\end{alignat}
with the global residual defined by $\mathbf{r}^k_{h} := \sum_{i=1}^s\mathbf{r}^k_{h,i}$. The equilibrated discretization flux for species $i$ is then given by
\begin{equation}\label{def:sigma_disc}
\sigmabo^{k}_{h,\textn{disc},i}:=- g_h^k\nabla s^k_{h,i}.
\end{equation}
The construction of the  micro-macro  misfit (remainder) flux for species $i$ is given component-wise by
\begin{align}\label{def:sigma_rem}
\sigmabo^{k}_{h,\textn{rem},i} :=
\frac{1}{d\tau_k}\Bigg[\int_{x_{i,1}}^{x_1}\mathbf{r}^k_{h,i}|_{(\zeta,x_2,\cdots,x_d)}\dzeta, 
\cdots,\int_{x_{i,d}}^{x_d}\mathbf{r}^k_{h,i}|_{(x_1,\cdots,x_{d-1},\zeta)}\dzeta \Bigg]^\top, 
\end{align}
Then,  for $1\leq k \leq M$, we let
\begin{equation}\label{def_discre_flux}
\sigmabo^{k}_{h}:=\sigmabo^{k}_{h,\textn{disc}}+\sigmabo^{k}_{h,\textn{rem}} := \sum_{i=1}^s \{\sigmabo^{k}_{h,\textn{disc},i}+\sigmabo^{k}_{h,\textn{rem},i}\}.
\end{equation}
The following key result shows that $\sigmabo_h^k$ from the above definition leads to an equilibrated flux in the sense of Definition~\ref{fluxequi}.
\begin{proposition} (Flux equilibration)\label{Prop:StrongEquilibration} Let the flux reconstruction $\sigmabo_h^k$, be defined by~\eqref{def_discre_flux}, where $\sigmabo^k_{h,\textn{disc},i}$ is defined by~\eqref{def:sigma_disc} and $\sigmabo^k_{h,\textn{rem},i}$  by~\eqref{def:sigma_rem}.  Then  $\sigmabo_h^k\in \Hdiv$ and we have the flux equilibration property \eqref{ass:equilibration} satisfied in the strong sense, i.e.,
\begin{align}\label{ass:equilibration2}
(1 - \tau_kr^k_h)s_h^k + \tau_k\Phi^k_{h}(s_h^k)s_h^k+ \tau_k\nabla \cdot \sigmabo_h^k = \tau_kf_{h}^k + s_h^{k-1}.
\end{align}
\end{proposition}
\begin{proof}
By construction, from relations \eqref{def:species} and \eqref{addlabel1}, we have
\bse
\begin{align}
\tau_k\nabla \cdot \sigmabo^k_{h,\textn{disc},i}&= \frac{1}{s}(\tau_k f^{k}_h + s_h^{k-1}) -(1 - \tau_kr_h^k)s^{k}_{h,i} -\tau_k\Phi_h^k(s_h^k)s^k_{h,i} - \mathbf{r}^k_{h,i}, \\
\tau_k\nabla \cdot \sigmabo^k_{h,\textn{rem},i} &= \mathbf{r}^k_{h,i}.
\end{align}
\ese
Thus, from \eqref{def_discre_flux} and \eqref{def:species} we obtain
\begin{align}
\nonumber \tau_k\nabla \cdot \sigmabo_h^k &= \tau_k f^{k}_h + s_h^{k-1} - \sum_{i = 1}^s \left\{(1 - \tau_k r_h^k)s^{k}_{h,i} + \tau_k\Phi_h^k(s_h^k)s^k_{h,i} \right\} \\
&= \tau_kf^{k}_h + s_h^{k-1} - (1 - \tau_kr_h^k)s^{k}_{h} - \tau_k\Phi_h^k(s_h^k)s^k_{h} \in L^2(\Omega).
\end{align}

\end{proof}
\begin{remark}[Equilibrated flux]
The choice of the equilibrated discretization flux in \eqref{def:sigma_disc} is motivated by the fact that we wish to be as general as possible with regards to from where $s_{h,i}^k$ is obtained. In practice, this means we may overestimate the error. In principle, one could solve an optimization problem to get a better estimate (see e.g., \cite{becker2001optimal}), but then additional requirements on $s_{h,i}^k$ is needed. 
\end{remark}
\subsection{Distinguishing the error components} 
The preceding developments lead to the following result.
\begin{theorem}[Error components]\label{thm:errorcomponents}
For the species $i\leq s$, let $s_{k,i}^k$ be the  population density as given in~\eqref{construct_pop_density} and the equilibrated flux  as characterized in~\eqref{ass:equilibration2}.    We have the following a posteriori error bound distinguishing the error components
For $1\leq k \leq M$, the following estimate holds true 
\begin{align}\label{finalestimate}
\left\{\vertiii{n^k - s_h^k}^2_k + \Jcal^k(n^k,s_h^k)\right\}^{\frac{1}{2}} \leq \sum_{i=1}^s  (\eta_{\textn{rem},i}^{k} + \eta_{r,i}^{k} + \eta_{g,i}^{k} + \eta_{\Phi,i}^{k}+\frac{1}{s}\etaosch^{k}), 
\end{align}
where 
\begin{equation}
 \eta_{\star,i}^{k} =\left\{\sum_{K\in \Zcal_h}  \left(\eta_{\star,K,i}^{k}\right)^2\right\}^{\frac{1}{2}},
\end{equation}
with    the species-dependent  counterparts of  the estimators~\eqref{def:loc_estimators1} defined by
\bse
\begin{align}
\etarhj^k &:= \tau_k\tilde{\omega}^{k}_{K}\norm{(r^k - r_h^k)s_{h,i}^k}_K,\\
\etaghj^k &:= \tau_kG_k^{-\frac{1}{2}}\norm{(g^k - g_{h}^k) \nabla s_{h,i}^k}_K, \\
\etaphihj^k &:= \tau_k\vertiii{(\Phi^k(s_{h,i}^k) - \Phi_{h}^k(s_{h,i}^k))s_{h,i}^k}_{k^\ast}, 
\end{align}
and the modeling-remainder  estimator defined by
\begin{equation}
 \eta_{\textn{rem},K,i}^k := \tau_kG_k^{-\frac{1}{2}}\norm{\sigmabo^k_{h,\textn{rem},i}}_K \label{etarem}.
\end{equation}
\ese
\end{theorem}
\begin{proof} 
First, observe that $\etaRK^{k} = 0$   since the flux satisfies the strong equilibration property \eqref{ass:equilibration2}. Next, substitute \eqref{def:species} and \eqref{def_discre_flux} in \eqref{first_estimate}. Due to \eqref{def:sigma_disc}, this gives $$\etaDF^k = \tau_kG_k^{-\frac{1}{2}}\norm{\sum_{i=1}^s\sigmabo^k_{h,\textn{rem},i}}_K.$$ Finally, use the triangle inequality to separate error components for individual species to arrive at \eqref{finalestimate}.
\end{proof}
\subsection{Detection of speciation events}
Using Theorem \ref{thm:errorcomponents} we can now detect speciation events in the SLM as follows: At each time step $k > 1$, we reconstruct the density function of each species in the system using the reconstruction operator, $\Dcal$, and compute the a posteriori error bound~\eqref{finalestimate}. Then, if the estimate exceeds some given tolerance we infer that a speciation event is about to happen for species $i$.
\section{Numerical examples}
\label{sec:ex}
In this section we present two numerical examples where we employ both the heuristic and multi-scale methods. The examples are chosen such that the first example conforms Assumption~\ref{assum:analysis}, and thus all the theoretical results apply. The first example thus allows us to validate the applicability of the methodology. The second example is motivated by the biological setting of speciation in a predator-prey setting, for which Assumption~\ref{assum:analysis} do not hold. Most notably, $d<3$ in addition to the operator $\Phi$ not satisfying the monotonicity condition~\eqref{monotonicity}, thus invalidating the assertion of Lemma~\ref{lem:edist}. This example provides numerical evidence to the efficacy of the multi-scale algorithm outside of the parameter space where we have been able to prove error bounds. In this context, the bound~\eqref{finalestimate} must be considered more loosely as an error indicator. 

For both examples we assess the accuracy of the speciation methods by solving simultaneously the PLM globally and calculating the statistical moments which are then compared to the corresponding species parameters. Moments of the reference PLM solution are always shown as dotted lines, and indicated by the letter $\mu$ in the legends. For the multi-scale method, the start and end of the speciation event is indicated by two vertical dotted lines. For the heuristic method, speciation is indicated by a single vertical dotted line. 

We assume the discretization errors associated with solving the models are negligible compared to the modeling error which we are interested in, hence in practice we calculate only the remainder estimator $\eta_{\textn{rem},i}^k$ from \eqref{etarem} when estimating~\eqref{finalestimate}. Moreover, the density reconstruction operator $\Dcal$ is implemented as a normal distribution, i.e., for a species $N_i^k = (n_i^k,x_i^k,\upsilon_i^k)$ we have
\begin{equation}
(\Dcal N_i^k)(x) :=  \frac{1}{\Lambda(\upsilon_i^k)}\exp\left(-\frac{1}{2}(x-x_i^k)^\top(\upsilon_i^k)^{-1}(x-x_i^k)\right).
\end{equation}
Note that this does not satisfy the homogenous Dirichlet boundary conditions. However, in practice for a large enough domain, boundary conditions are still satisfied within working precision. An alternative approach would be to truncate the tails of the normal distribution, but for the problems considered herein, we do not expect this to make any difference in the results. Finally, the $m$-species compression operator, $\Ccal^m_{B_{i,\nu}^k}$, is implemented using a nonlinear least squares iteration (here for $m = 2$), where we iterate until convergence with a relative tolerance of $1e-3$. 
\subsection{Numerical approximations} 
We advance in time using a fourth order Runge-Kutta scheme (RK4). Numerical integration is by the midpoint rule, and the differential term of the PLM is approximated by the two-point flux approximation (TPFA) method on a regular Cartesian grid. All numerical examples are implemented in MatLab v. R2019b. 
\subsection{Example 1}
For the first numerical example we let $d = 3$, and consider a system initially consisting of $s=1$ species, and where one speciation event occurs. Here, we let the micro and macro time increments be given by $\tau_m = \tau_M = 5e-2$, respectively, and set $T=600$ as the final time. Spatial grid size is $\Delta x = (1,1,1)/20$. The residual relative tolerance is chosen as $\tol_{\textn{res}} = 5e1$, and the number of standard deviations for the trait space density regions (and tolerance distance between mean trait coordinates during speciation) as $\nu = 10$. When the computed error bound exceeds the tolerance, we backtrack 100 time units before initiating the multi-scale/heuristic algorithms. Furthermore, since the Assumption~\ref{assum:analysis} is fulfilled for this example, we calculate the difference between the global PLM solution and the reconstructed SLM solution in the energy norm, i.e., $\vertiii{n^k - s^k}_k$. The trait space domain is the unit cube, i.e., $\Omega = [0,1]^3$. Initial data is given by $N_0 = (n_{0}, x_{0}, \upsilon_0)$, where
\begin{equation}
n_{0} = 2e-1, \quad  x_{0} = (0.2,0.2,0.2), \quad \textn{ and } \quad \upsilon_{0} = 5e-3 \times \id.
\end{equation}
\subsubsection{Parameters}
We impose a speciation event on the system by having a time dependent growth-rate, where a single attractor point in trait space gradually transitions into two attractor points. In particular, for $\gamma > 0$ the growth rate is defined as
\begin{equation}
r(x,t) := 1 - \gamma(f_0(t)r_0(x) + f_{\infty}(t)r_{\infty}(x)),
\end{equation}
where 
\bse
\begin{align}
r_0(x) &:= \norm{x - (0.5,0.5,0.5)}^2, \label{attr1}\\
r_{\infty}(x) &:= \norm{x - (0.2,0.8,0.8)}^2\norm{x - (0.8,0.2,0.2)}^2,\label{attr2}
\end{align}
\ese
and where (for $\theta > 0$)
\begin{equation}
f_0(t) := 
\begin{cases}
1 - t/\theta, &t < \theta, \\
0, &t \geq \theta,
\end{cases}
\quad \textnormal{ and } \quad 
f_\infty(t) := 
\begin{cases}
t/\theta, &t < \theta, \\
1, &t \geq \theta.
\end{cases}
\end{equation}
Thus, $\gamma$ determines the speed at which the species moves towards the attractor points set by \eqref{attr1} or \eqref{attr2}, and $\theta$ is the transition time between $r_0$ and $r_{\infty}$. For the present situation we choose the following parameter values
\begin{equation}
\gamma = 2, \quad \textn{ and } \quad \theta = 300.
\end{equation}
The remaining coefficients are defined by constant values, i.e.,
\bse
\begin{align}
b(x,t) &= 1e-3,\\
\alpha(x, y,t) &= -1,\\
g(x,t) &= 5e-6 \times \id.
\end{align}
\ese
\subsubsection{Simulation}
With the parameters given in the previous section, we employ both the heuristic and multi-scale speciation algorithms. Figure~\ref{ex1:abundance} below shows the species abundance, figure~\ref{ex1:x1} the species mean trait coordinates, figure~\ref{ex1:variance} the largest eigenvalue of the trait covariance matrix, figure~\ref{ex1:residuals} the a posteriori error bound, and figure~\ref{ex1:edist} the error in energy norm, as functions of time. Note that due to the presence of $G^{-\frac{1}{2}}$ in the definition of $\eta_{rem,i}^k$ (eq. \eqref{etarem}), the magnitudes in figure~\ref{ex1:residuals} can not be directly compared to those of the preceding figures.

We observe from the results a very close match between the reference solution and the multi-scale method throughout the simulation time. Indeed, the error in the energy norm, as seen in Figure~\ref{ex1:edist} is lower during the multi-scale window than during the pure SLM simulation, indicating that the modeling error of the SLM in terms of capturing species dynamics dominates over the error associated with the speciation event. In terms of species-level parameters, as shown in the remaining figures, the qualitative match is also quite satisfactory for the multi-scale mathod outside the speciation event (these quantities are of course not defined during the speciation event itself). 

In contrast, while the heuristic method is somewhat acceptable in terms of capturing the initial dynamics and final state, it does not capture the speciation dynamics themselves as accurately as the multi-scale method, thus emphasizing the value of the multi-scale simulation framework. The lack of accuracy during the speciation event leads to later errors in the higher moments of the solution (i.e., velocities and covariances in trait space), as seen in particular in Figures \ref{ex1:x1} and \ref{ex1:variance}.

\begin{figure}[h]
  \centering
\begin{subfigure}{.45\textwidth}
  \centering
  \includegraphics[width=1.\linewidth]{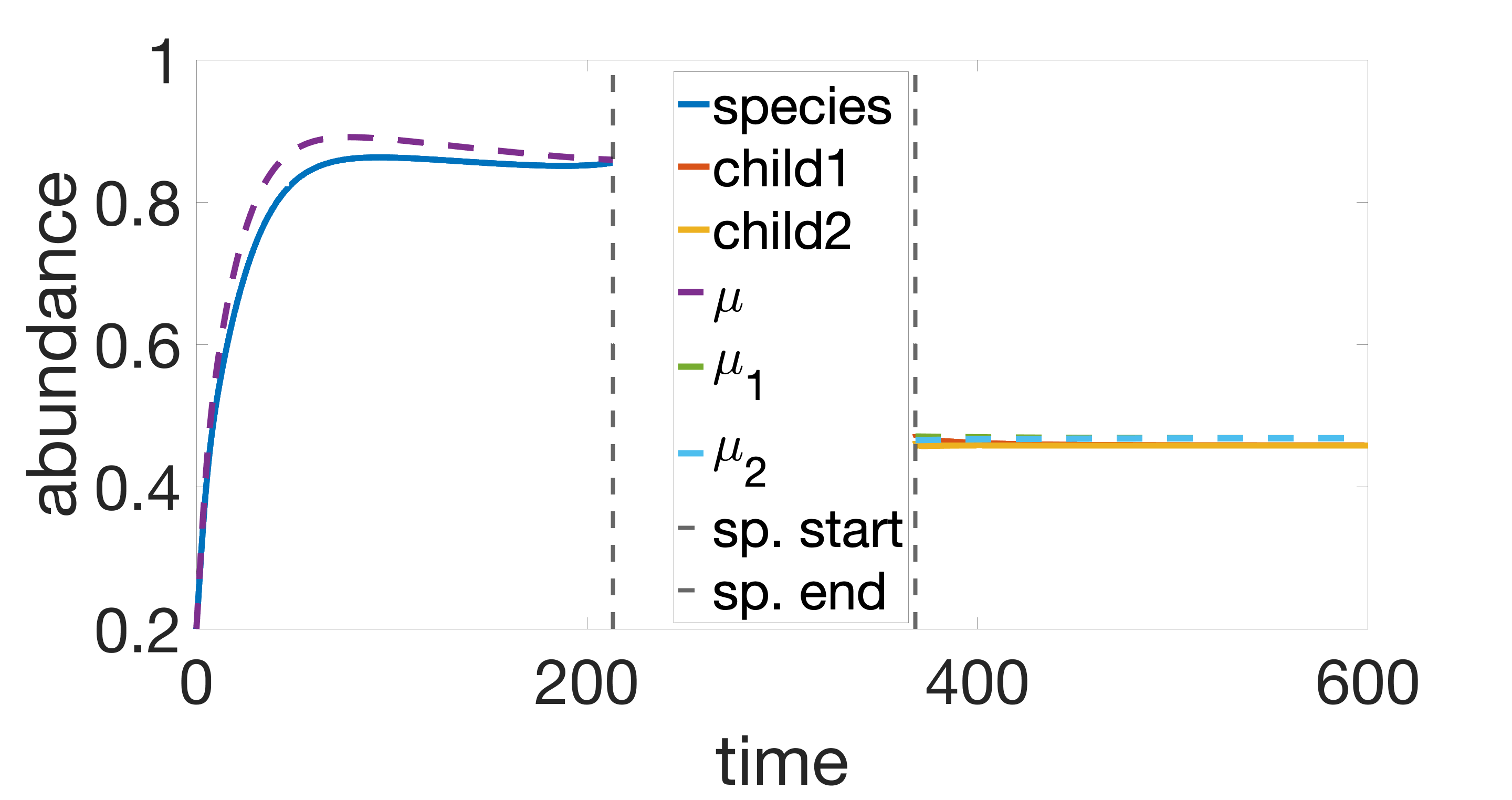}
  \caption{Multi-scale method.}
\end{subfigure}%
\hspace{.05\textwidth}
\begin{subfigure}{.45\textwidth}
  \centering
  \includegraphics[width=1.\linewidth]{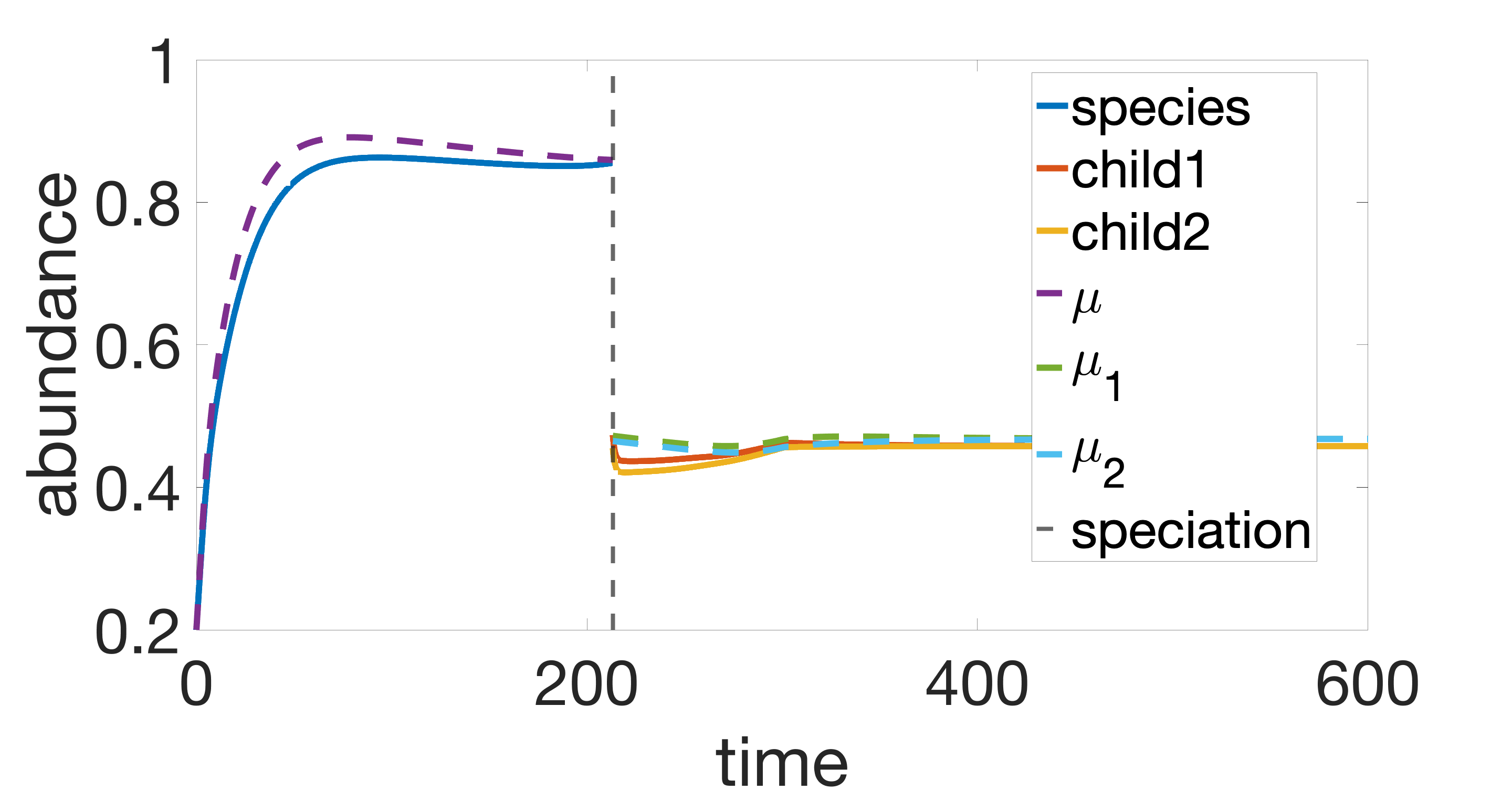}
  \caption{Heuristic method.}
\end{subfigure}  
  \caption{Species abundance as functions of time.}
  \label{ex1:abundance}
\end{figure}

\begin{figure}[h]
  \centering
\begin{subfigure}{.45\textwidth}
  \centering
  \includegraphics[width=1.\linewidth]{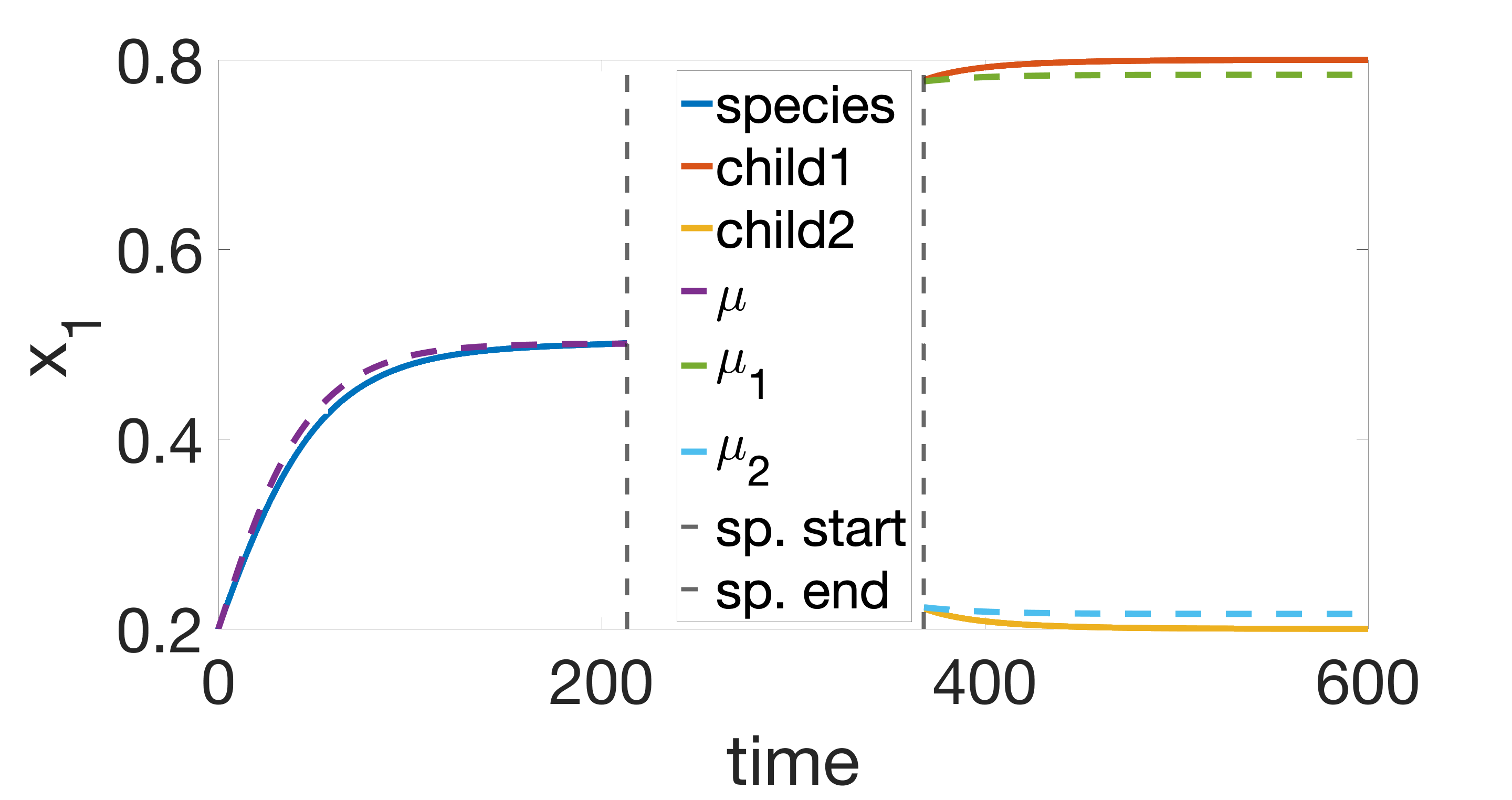}
  \caption{Multi-scale method.}
\end{subfigure}%
\hspace{.05\textwidth}
\begin{subfigure}{.45\textwidth}
  \centering
  \includegraphics[width=1.\linewidth]{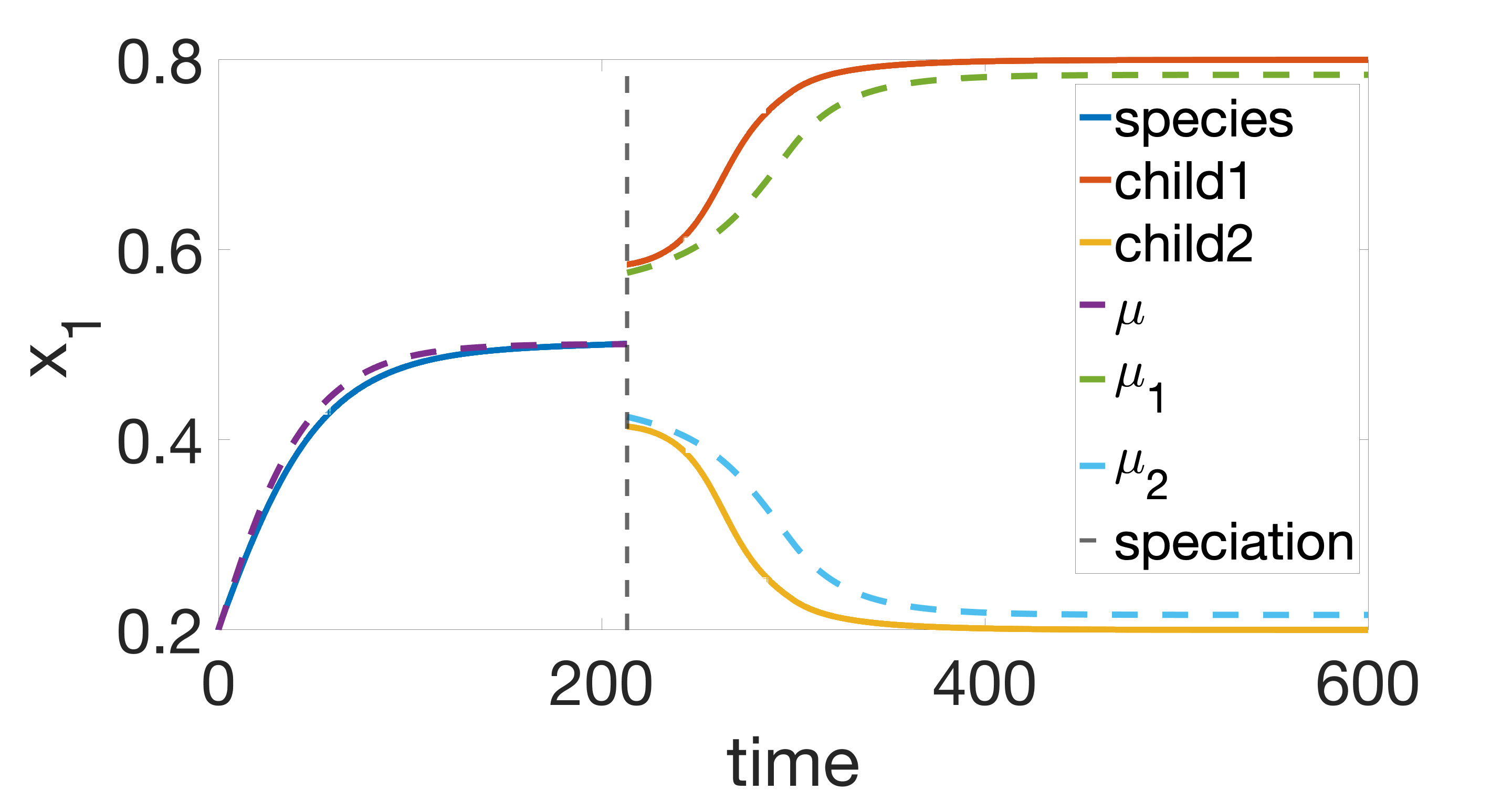}
  \caption{Heuristic method.}
\end{subfigure}  
  \caption{Species mean traits as functions of time (first component).}
  \label{ex1:x1}
\end{figure}

\begin{figure}[h]
  \centering
\begin{subfigure}{.45\textwidth}
  \centering
  \includegraphics[width=1.\linewidth]{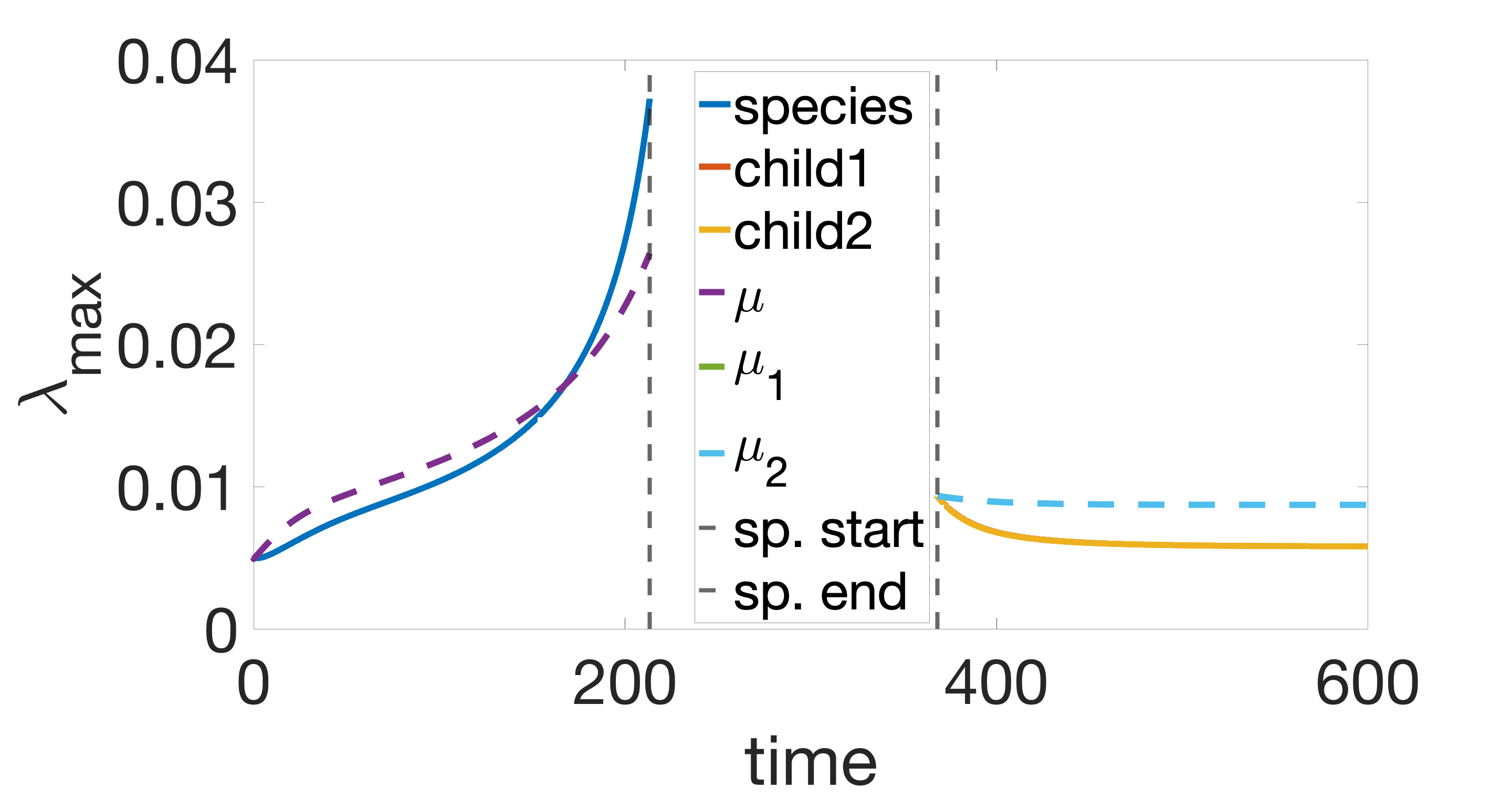}
  \caption{Multi-scale method.}
\end{subfigure}%
\hspace{.05\textwidth}
\begin{subfigure}{.45\textwidth}
  \centering
  \includegraphics[width=1.\linewidth]{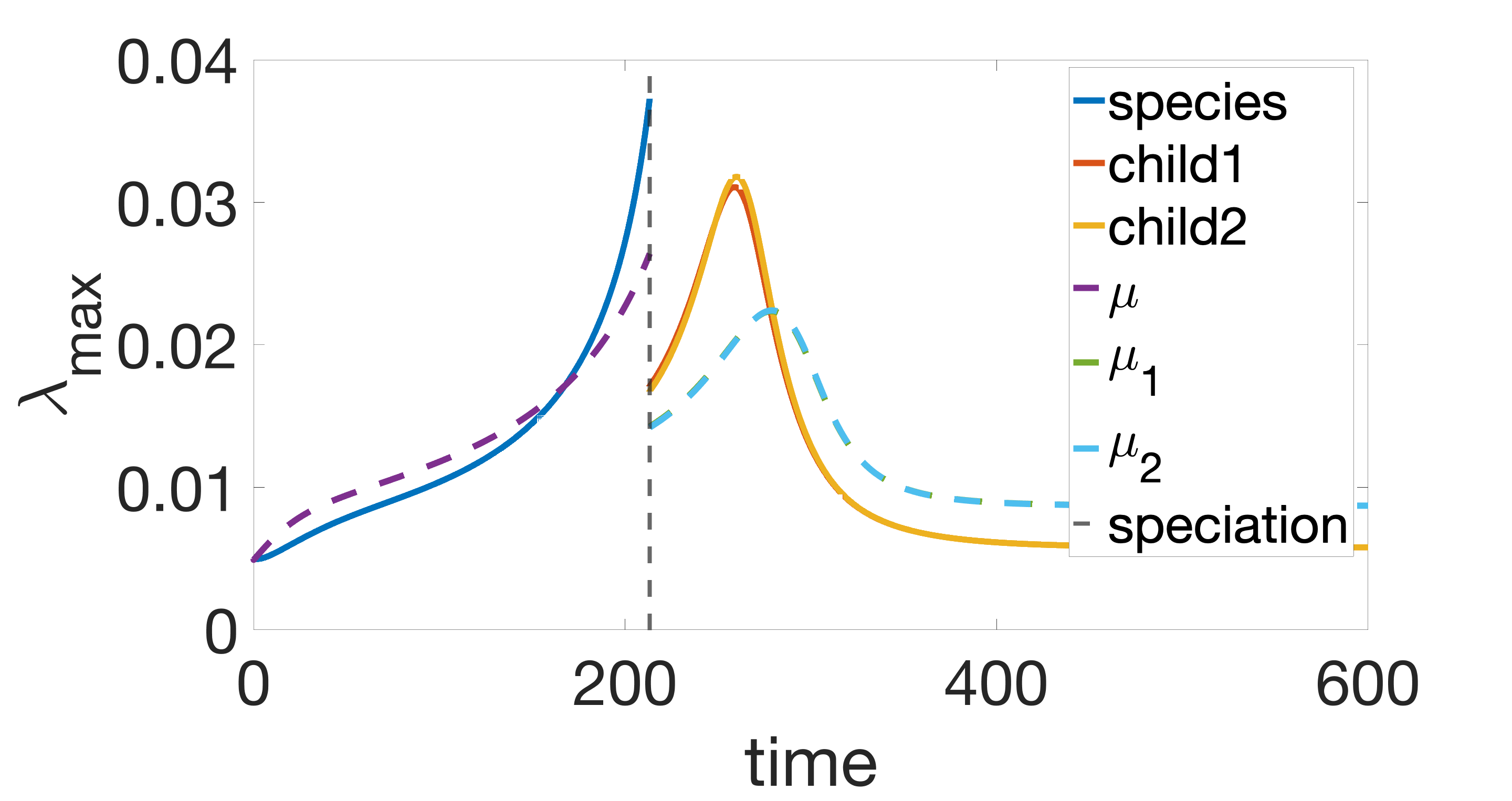}
  \caption{Heuristic method.}
\end{subfigure}  
  \caption{Maximum eigenvalue of trait covariance matrix as functions of time.}
  \label{ex1:variance}
\end{figure}

\begin{figure}[h]
  \centering
\begin{subfigure}{.45\textwidth}
  \centering
  \includegraphics[width=1.\linewidth]{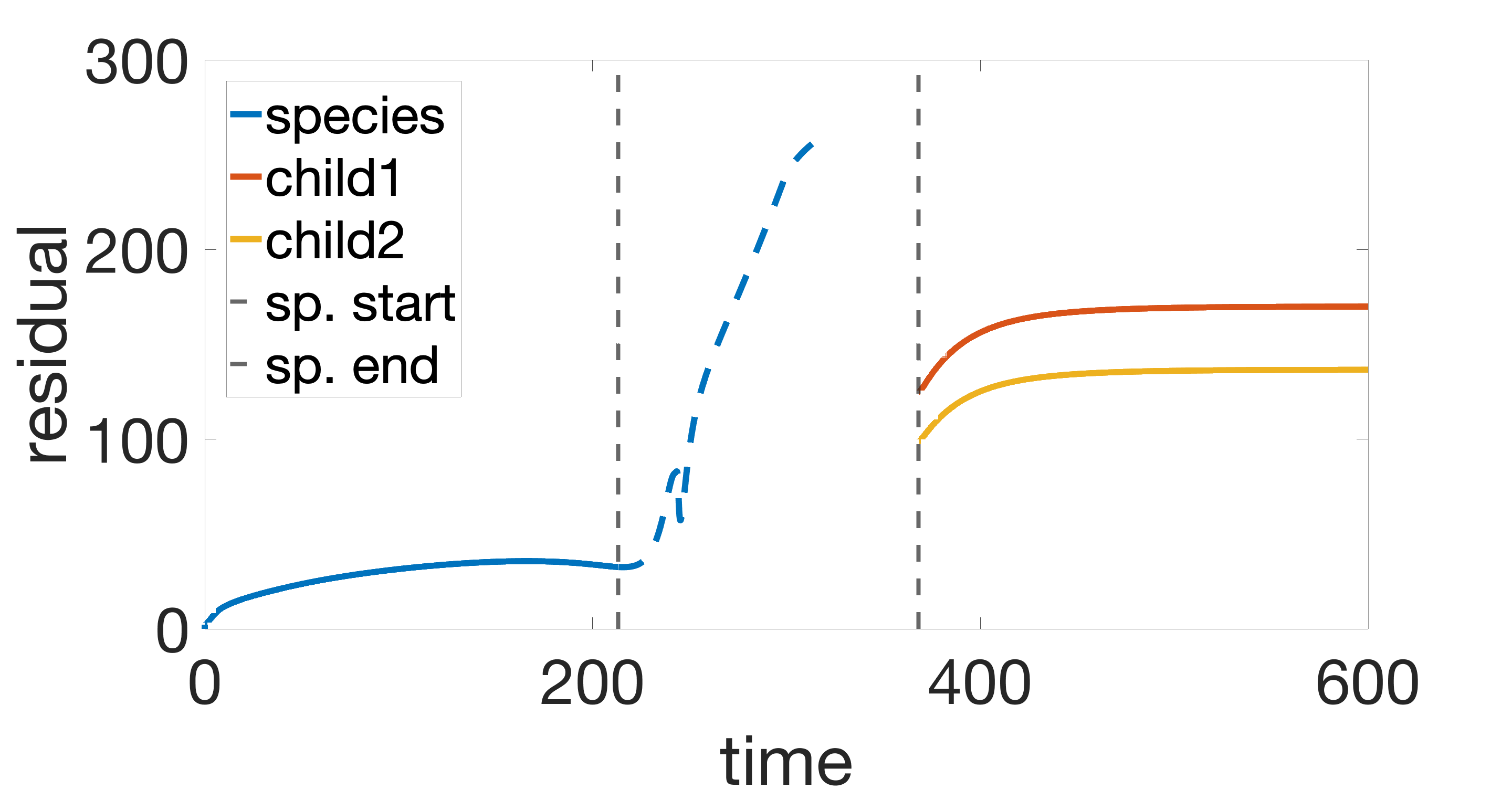}
  \caption{Multi-scale method.}
\end{subfigure}%
\hspace{.05\textwidth}
\begin{subfigure}{.45\textwidth}
  \centering
  \includegraphics[width=1.\linewidth]{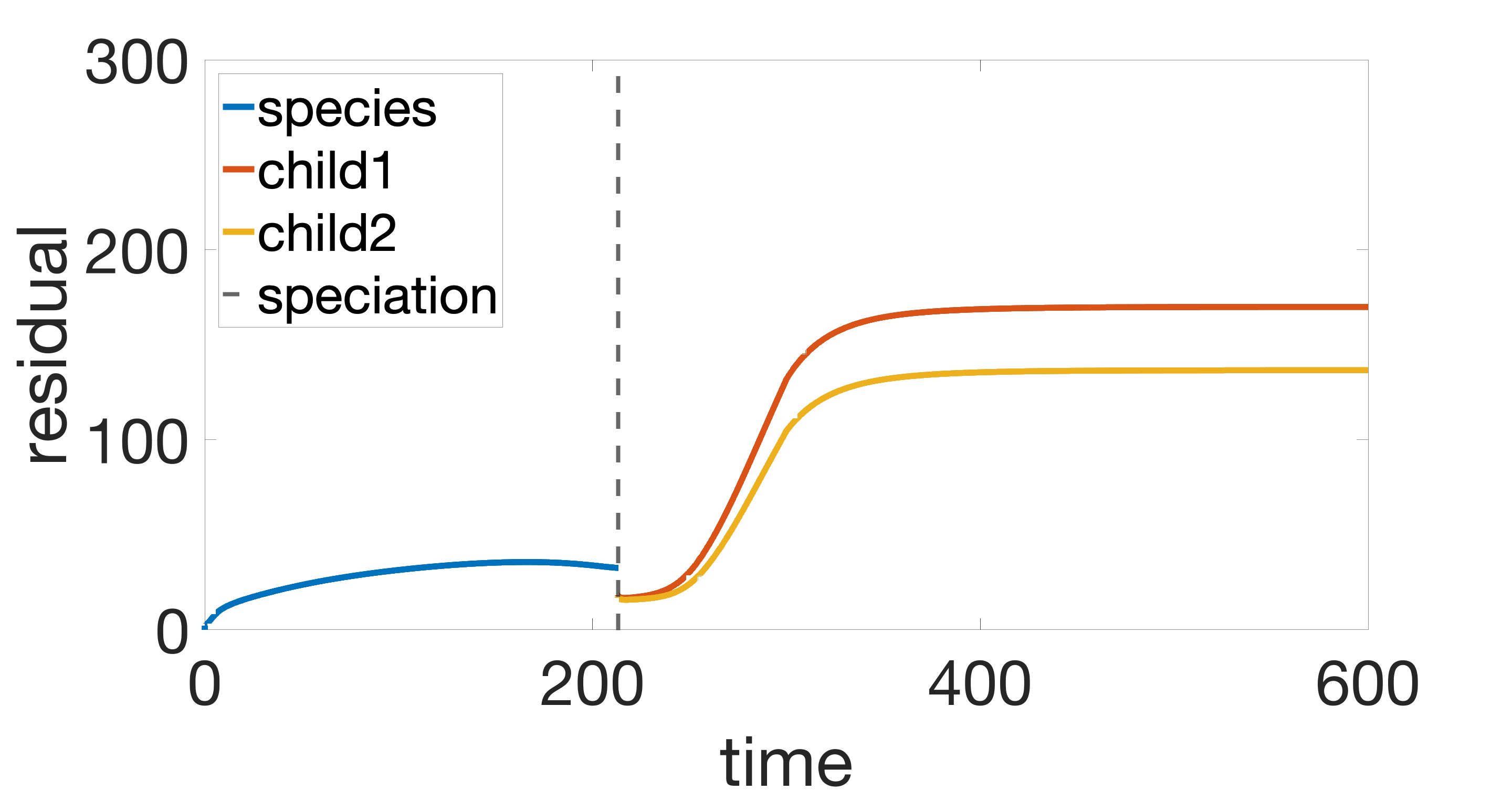}
  \caption{Heuristic method.}
\end{subfigure}  
  \caption{A posteriori modeling-remainder estimator, $\eta_{\textn{rem},i}^{k}$, as a function of time.}
  \label{ex1:residuals}
\end{figure}

\begin{figure}[h]
  \centering
\begin{subfigure}{.45\textwidth}
  \centering
  \includegraphics[width=1.\linewidth]{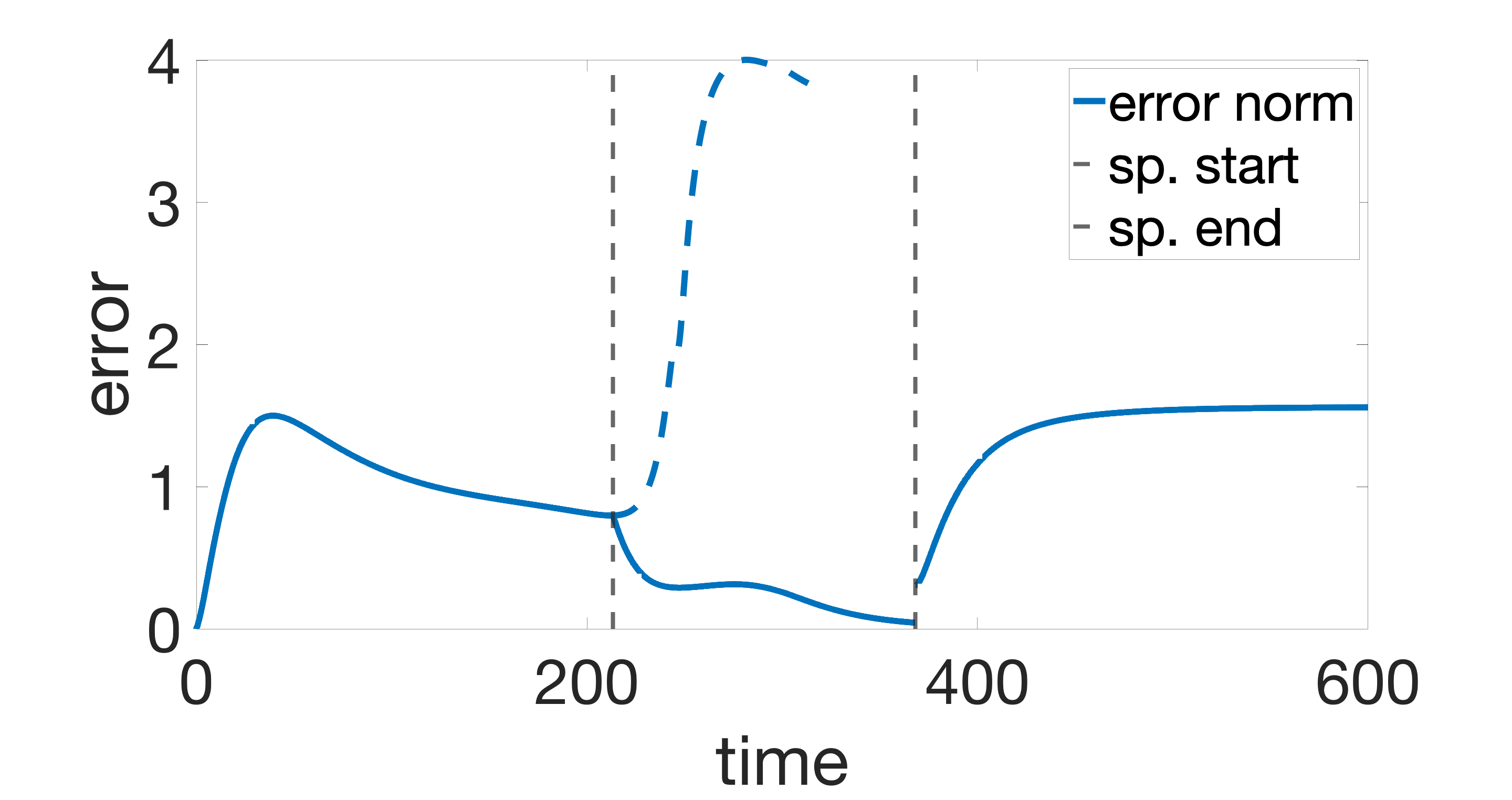}
  \caption{Multi-scale method.}
\end{subfigure}%
\hspace{.05\textwidth}
\begin{subfigure}{.45\textwidth}
  \centering
  \includegraphics[width=1.\linewidth]{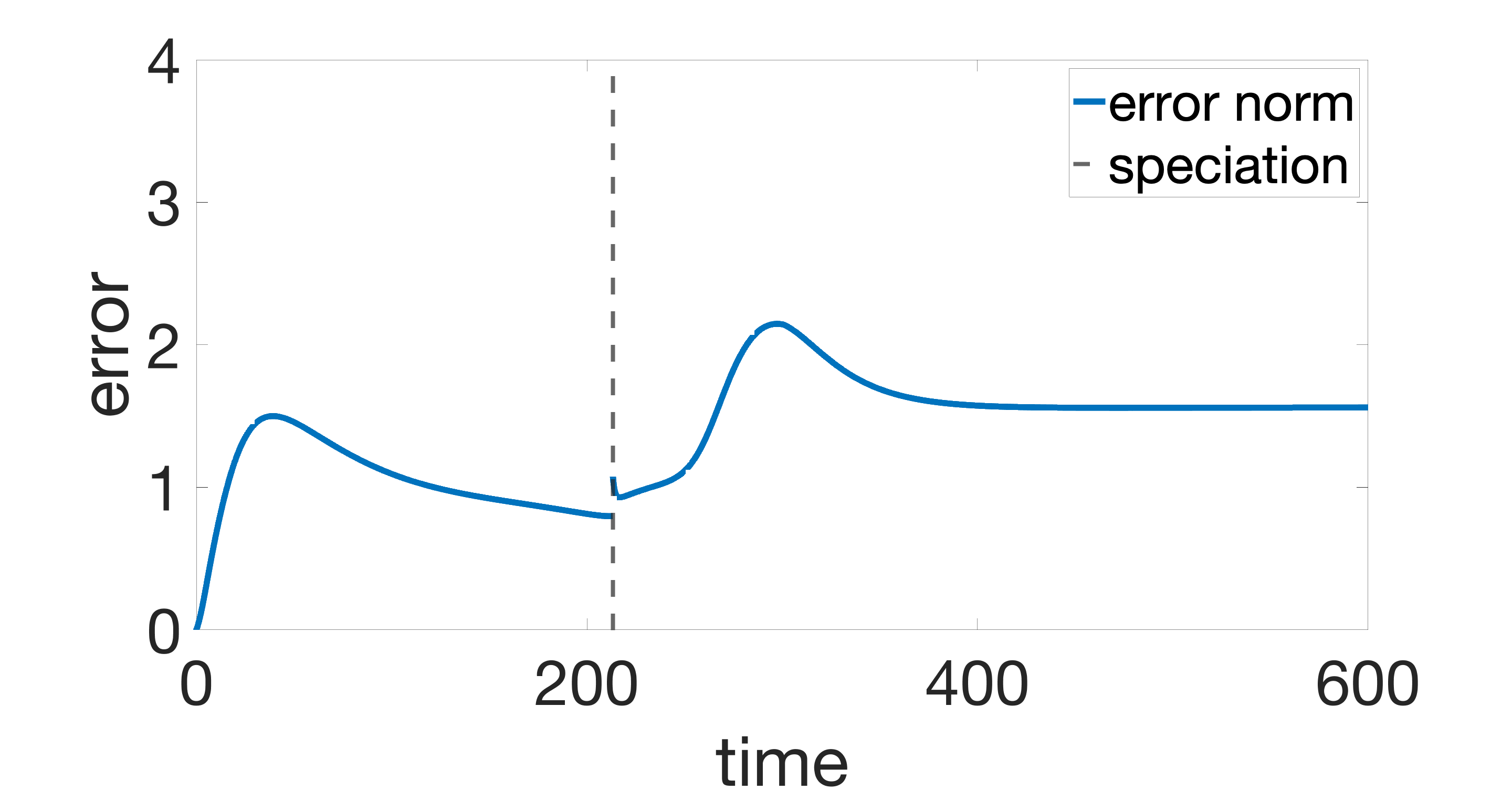}
  \caption{Heuristic method.}
\end{subfigure}  
  \caption{Global error measured in energy norm as function of time.}
  \label{ex1:edist}
\end{figure}
\subsection{Example 2}
For the second numerical example we let $d=2$, and consider a predator-prey system initially consisting of $s=2$ species, and where the `prey'-species undergoes a speciation event by traveling along a Y-shaped ridge in trait space. Hence, the divergence will happen at the branching point of the ridge. Here, we let the micro and macro time increments be given by $\tau_m = 1e-3$ and $\tau_M = 1e-2$, respectively, and set $T=600$ as the final time. Spatial grid size is $\Delta x = (1,1)/50$. The residual relative tolerance is chosen as $\tol_{\textn{res}} = 1e1$, and the number of standard deviations for the trait space density regions as $\nu = 10$. When the computed error bound exceeds the tolerance, we backtrack 250 time units before initiating the multi-scale/heuristic algorithms. The trait-space domain consists of two disjoint regions; $\Omega_1 = [0,1]^2$ and $\Omega_2 = [2,3]^2$, i.e., $\Omega = \Omega_1 \cup \Omega_2$, wherein the two species are located, respectively. Initial data is given by $N_0 = (N_{\textn{prey},0},N_{\textn{pred},0})$, where
\bse
\begin{align}
&n_{\textn{prey},0} = 2e-1, \quad  x_{\textn{prey},0} = (0.5,0.3), \quad \textn{ and } \quad \upsilon_{\textn{prey},0} = 5e-3 \times \id,\\
&n_{\textn{pred},0} = 2e-1, \quad  x_{\textn{pred},0} = (2.5,2.5), \quad \textn{ and } \quad \upsilon_{\textn{pred},0} = 5e-3 \times \id.
\end{align}
\ese
\subsubsection{Parameters}
We impose a speciation event upon the `prey'-species by initializing it at the foot of a ridge in trait space, and as the species travels along this ridge, the ridge splits into two branches. This branching ridge is incorporated in the growth-rate coefficient, which for $c_0, \delta, r_0 > 0$ and $\vecc \in \real^2$, is defined as
\begin{equation}
r(x) := 
\begin{cases}
c_0 + \delta \vecc \cdot x^\top - \varphi_{\epsilon}(x) \ast \dist(x, Y), \quad &x \in \Omega_1, \\
-r_0, &x \in \Omega_2,
\end{cases}
\end{equation}
where the set of points $Y\subset \Omega_1$, is the three line segments connecting the nodes $\{(0.5,0.3), (0.5,0.5), (0.2,0.7), (0.8,0.7)\}$ to form a Y-shape, and where $\varphi_\epsilon$ is the mollifier function centered on the origin with radius of support $\epsilon > 0$, and where $\dist(x, Y)$ is the distance from the point $x$ to the set $Y$. Hence, $c_0$ is the growth-rate of the prey, $\delta$ is the speed at which it travels along the ridge, $\vecc$ is a direction vector, $\epsilon$ is the steepness of the ridge, and $r_0$ is the loss rate of the predator. The interaction coefficient is defined by
\begin{equation}
\alpha(x,y) := 
\begin{cases}
0, \quad &x \in \Omega_1, y \in \Omega_1, \\
- \gamma, \quad &x \in \Omega_1, y \in \Omega_2, \\
\beta, \quad &x\in\Omega_2, y \in \Omega_1, \\
0, \quad &x \in \Omega_2, y \in \Omega_2,
\end{cases}
\end{equation}
where $\gamma > 0$ is the rate of predation upon the prey, and where $\beta > 0$ is the growth rate of the predator. For the present situation we choose the following parameter values
\begin{align}
\nonumber c_0 = 0, \quad \delta &= 0.8, \quad r_0 = 0.5, \quad \epsilon = 0.2, \quad \gamma = 3.0, \\
\beta &= 8.0, \text{ and } \vecc = (0,1)^\top.
\end{align}
The remaining coefficients are defined by constant values, i.e.,
\bse
\begin{align}
b(x,t) &:= 0, \\
g(x,t) &:= 2e-6\times\id.
\end{align}
\ese
\subsubsection{Simulation}
With the parameters given in the previous section, we employ both the heuristic and multi-scale algorithms. Figures~\ref{ex2:prey_abundance}--\ref{ex2:pred_abundance} below shows the species abundance, figures~\ref{ex2:prey_x1}--\ref{ex2:prey_x2} the species mean trait coordinates, figure~\ref{ex2:prey_variance} the largest eigenvalue of the trait covariance matrix, and figures~\ref{ex2:prey_residuals}--\ref{ex2:pred_residuals} the a posteriori error bounds, as functions of time. 


We recall that for this example, Assumption~\ref{assum:analysis} is not satisfied, and the results are thus not expected to be as strong as for the previous example. Indeed, we observe that while the multi-scale method performs fairly overall, the a posteriori error bound, now being only an error indicator, is not precise enough to identify the speciation event early enough, necessitating a larger backtrack window than the previous example. Note also that the abundance plots, figures \ref{ex2:prey_abundance}--\ref{ex2:pred_abundance}, indicate that the predator-prey cycles have been shifted out of phase, but the correct structure is still retained. However, the multi-scale method still clearly outperforms the heuristic method, which due to the reliance on the same error indicator also suffers from a somewhat delayed speciation event. On the other hand, the heuristic method models the speciation event less accurately than the multi-scale method, and thus the errors after the speciation event are significantly larger when seen in terms of the species-level parameters. 

\begin{figure}[h]
  \centering
\begin{subfigure}{.45\textwidth}
  \centering
  \includegraphics[width=1.\linewidth]{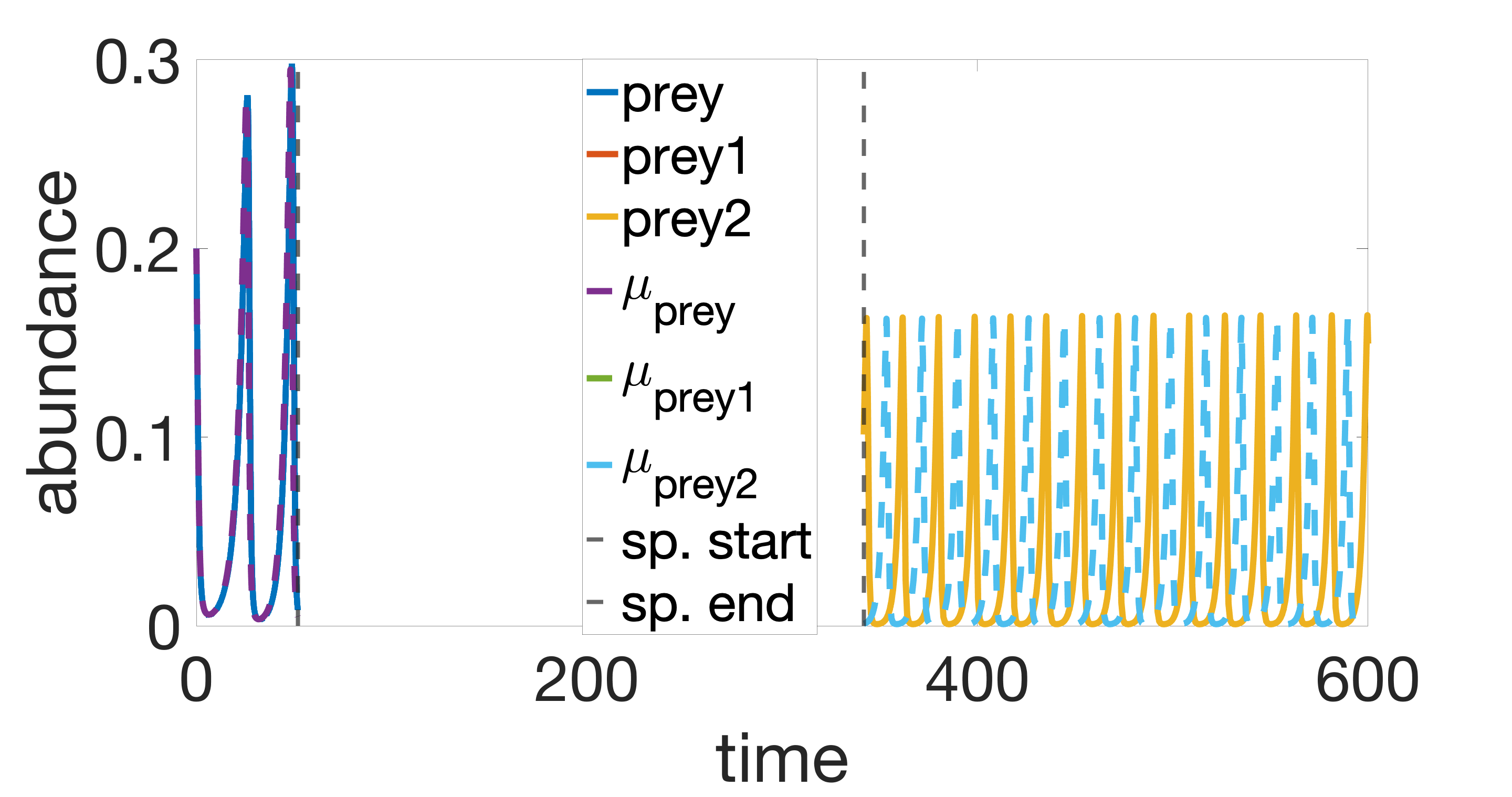}
  \caption{Multi-scale method.}
\end{subfigure}%
\hspace{.05\textwidth}
\begin{subfigure}{.45\textwidth}
  \centering
  \includegraphics[width=1.\linewidth]{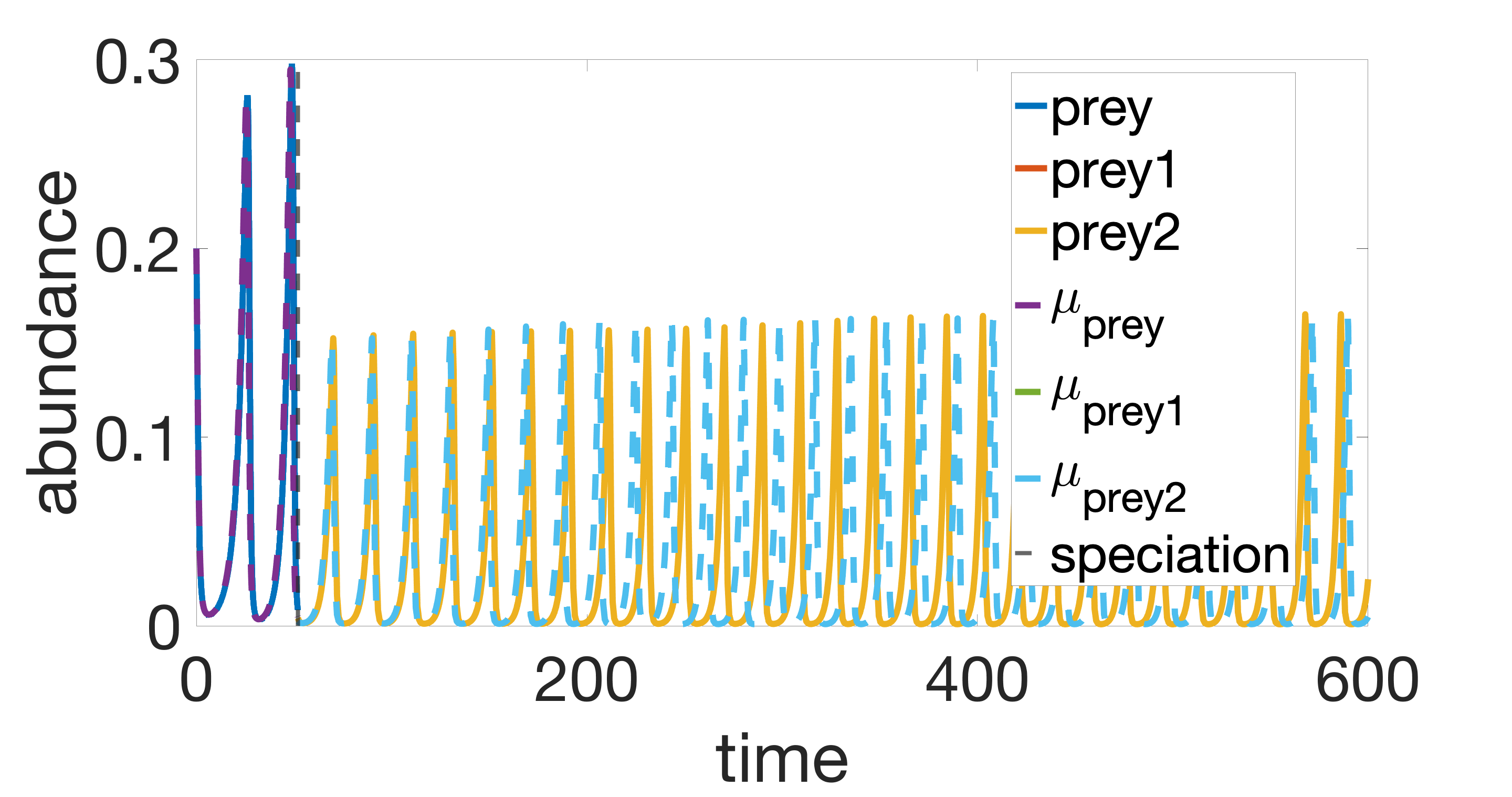}
  \caption{Heuristic method.}
\end{subfigure}  
  \caption{Prey abundance as function of time. Curves for child-species are overlapping.}
  \label{ex2:prey_abundance}
\end{figure}

\begin{figure}[h]
  \centering
\begin{subfigure}{.45\textwidth}
  \centering
  \includegraphics[width=1.\linewidth]{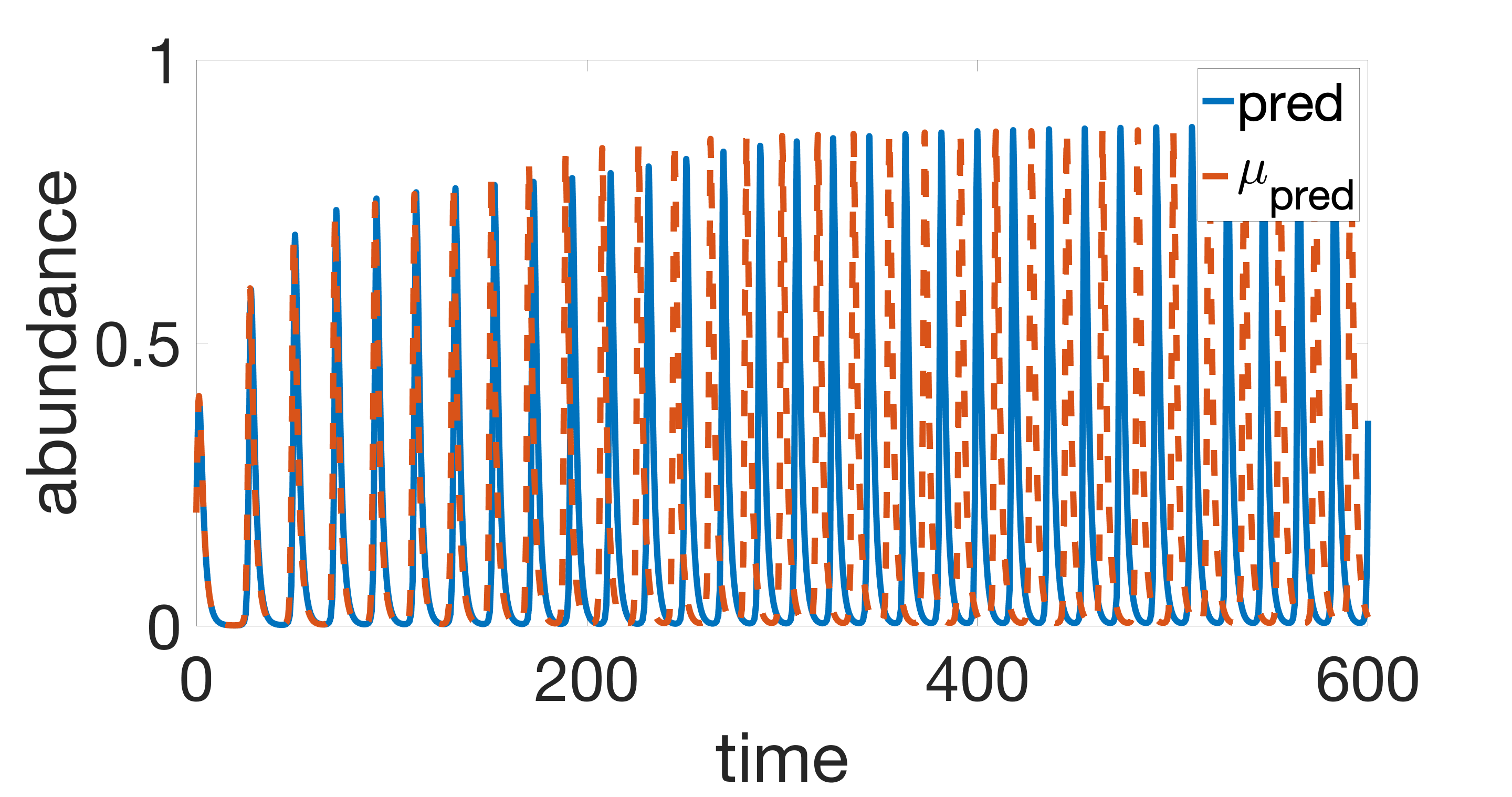}
  \caption{Multi-scale method.}
\end{subfigure}%
\hspace{.05\textwidth}
\begin{subfigure}{.45\textwidth}
  \centering
  \includegraphics[width=1.\linewidth]{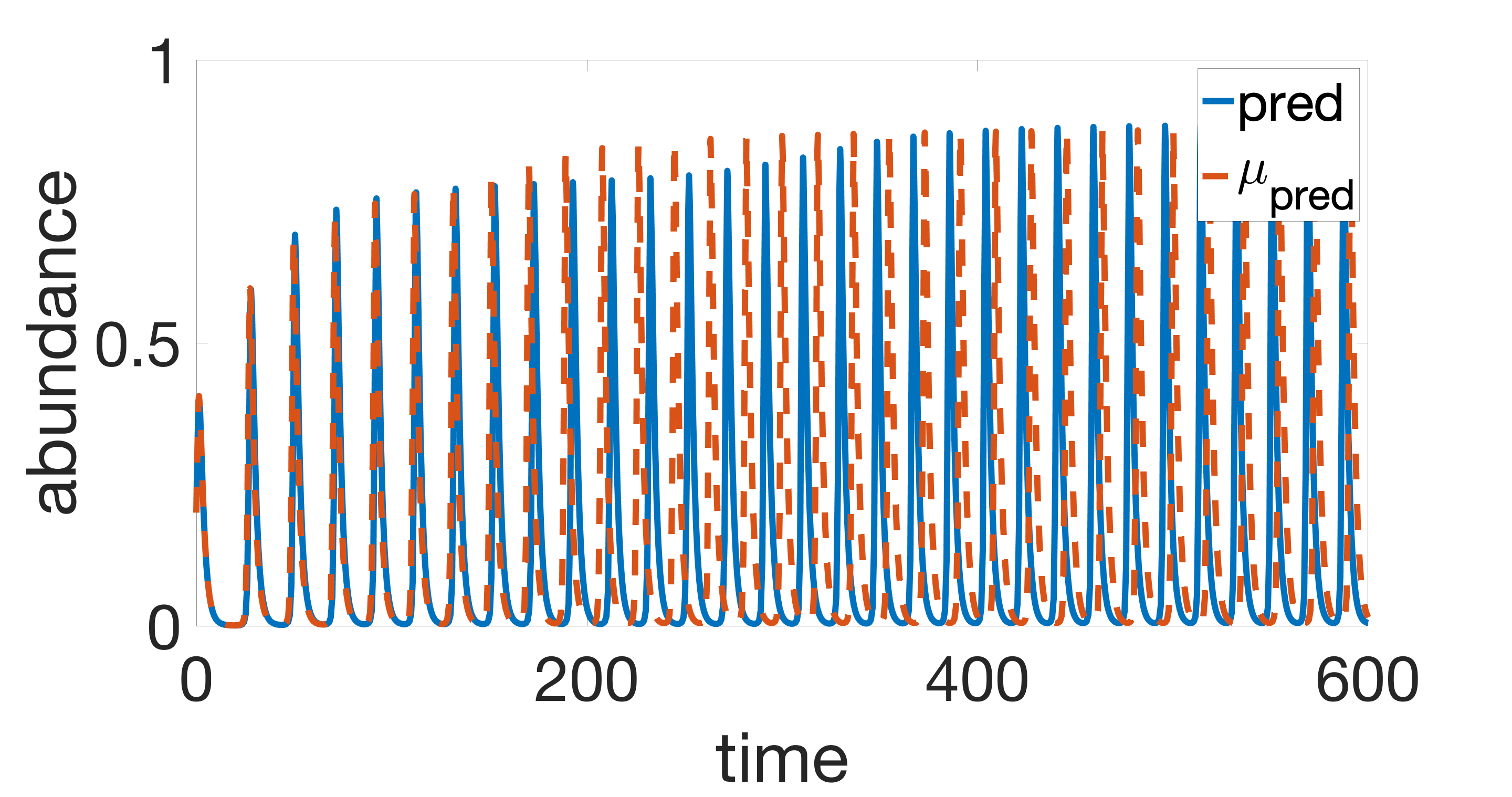}
  \caption{Heuristic method.}
\end{subfigure}  
  \caption{Predator abundance as function of time. }
  \label{ex2:pred_abundance}
\end{figure}

\begin{figure}[h]
  \centering
\begin{subfigure}{.45\textwidth}
  \centering
  \includegraphics[width=1.\linewidth]{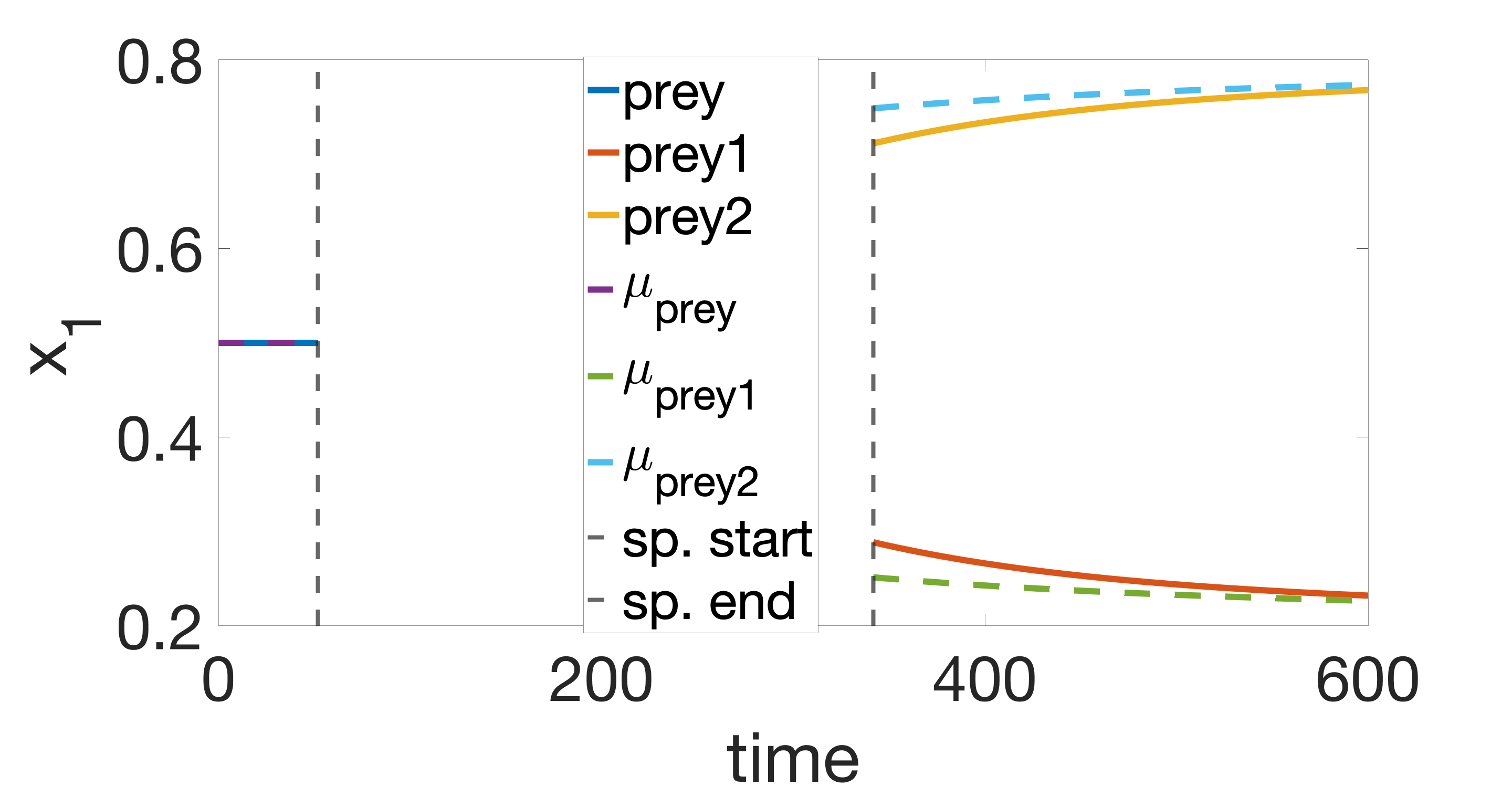}
  \caption{Multi-scale method.}
\end{subfigure}%
\hspace{.05\textwidth}
\begin{subfigure}{.45\textwidth}
  \centering
  \includegraphics[width=1.\linewidth]{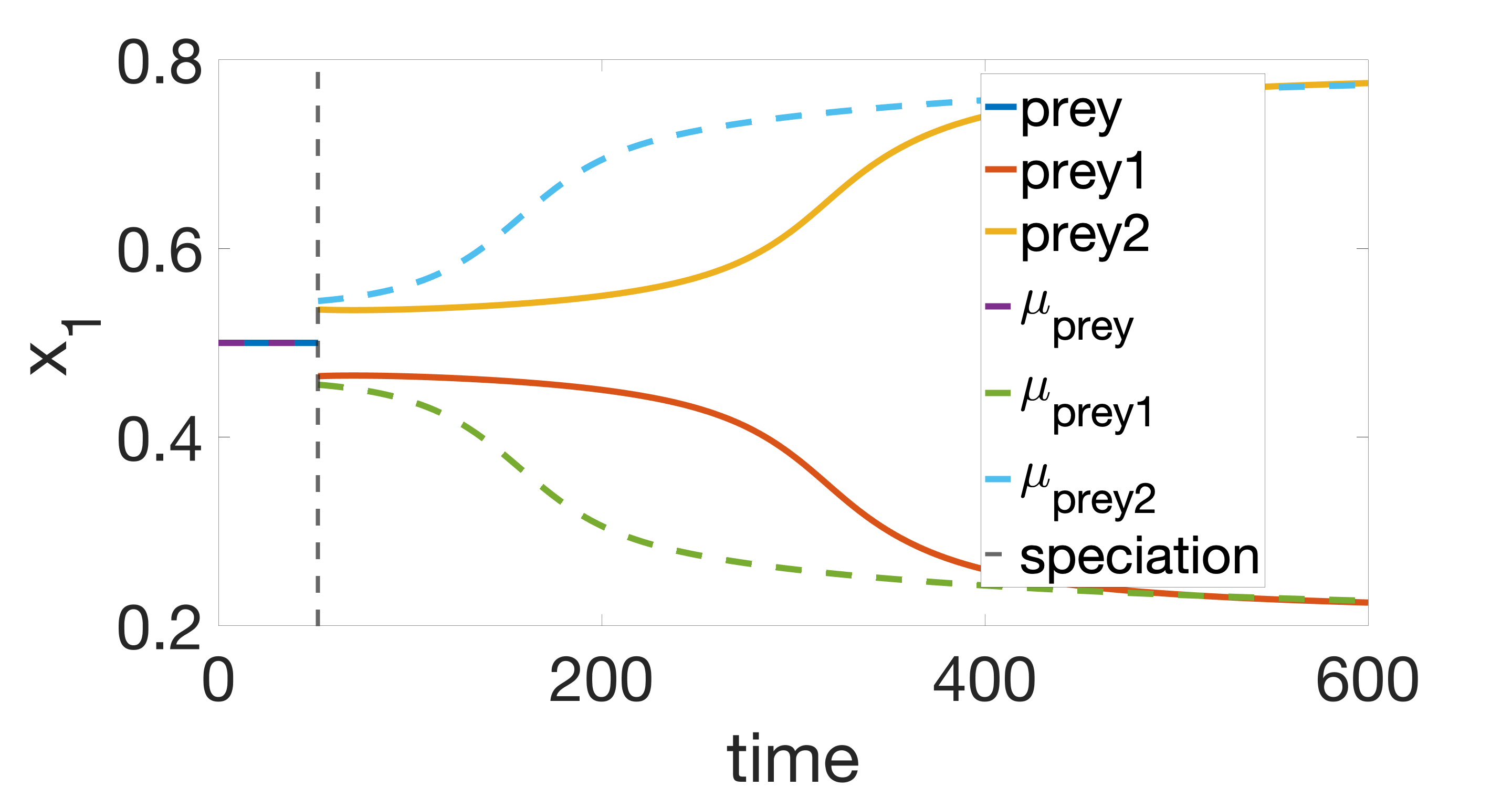}
  \caption{Heuristic method.}
\end{subfigure}  
  \caption{Prey mean traits coordinate as functions of time (first component).}
  \label{ex2:prey_x1}
\end{figure}


\begin{figure}[h]
  \centering
\begin{subfigure}{.45\textwidth}
  \centering
  \includegraphics[width=1.\linewidth]{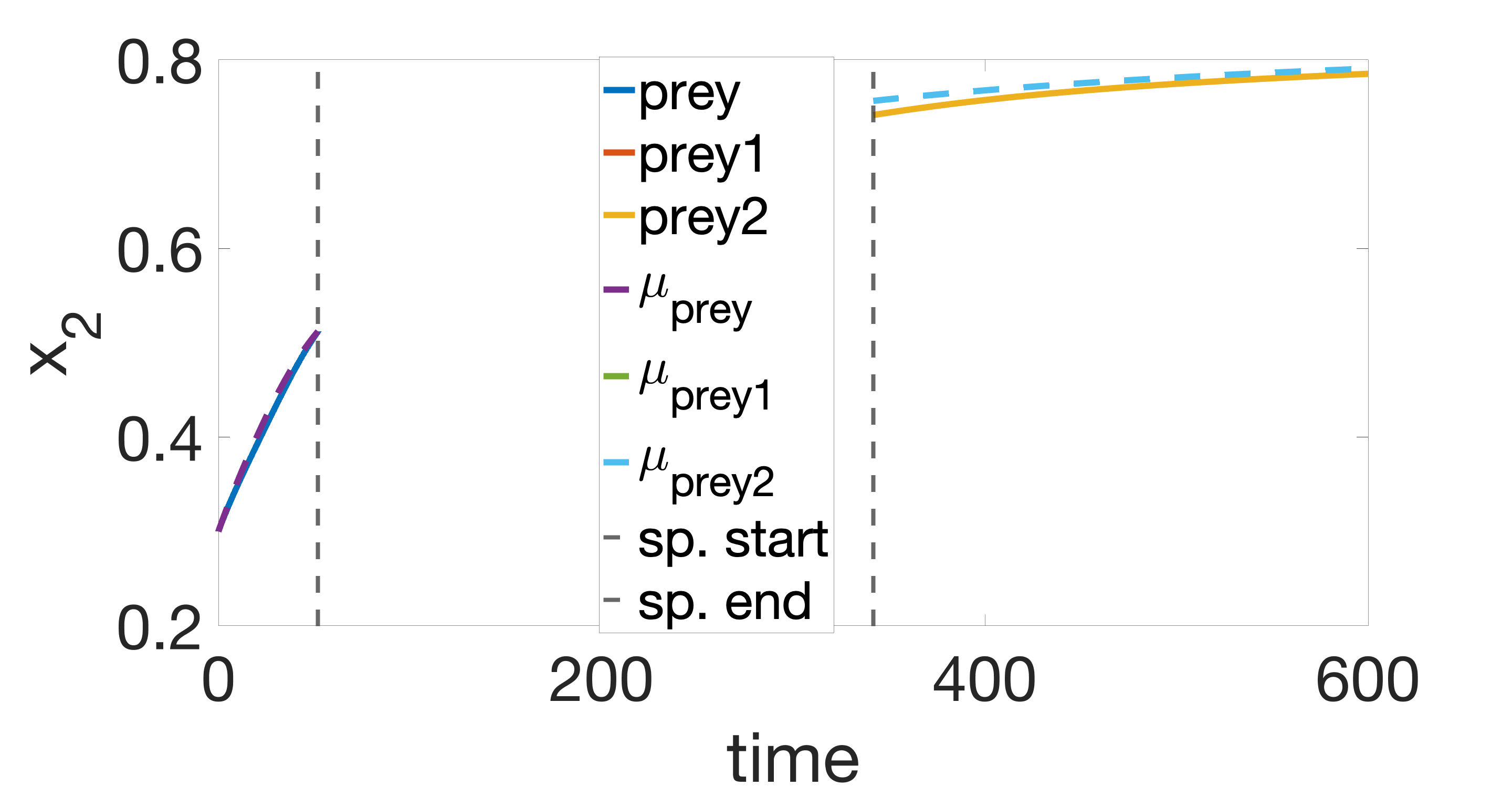}
  \caption{Multi-scale method.}
\end{subfigure}%
\hspace{.05\textwidth}
\begin{subfigure}{.45\textwidth}
  \centering
  \includegraphics[width=1.\linewidth]{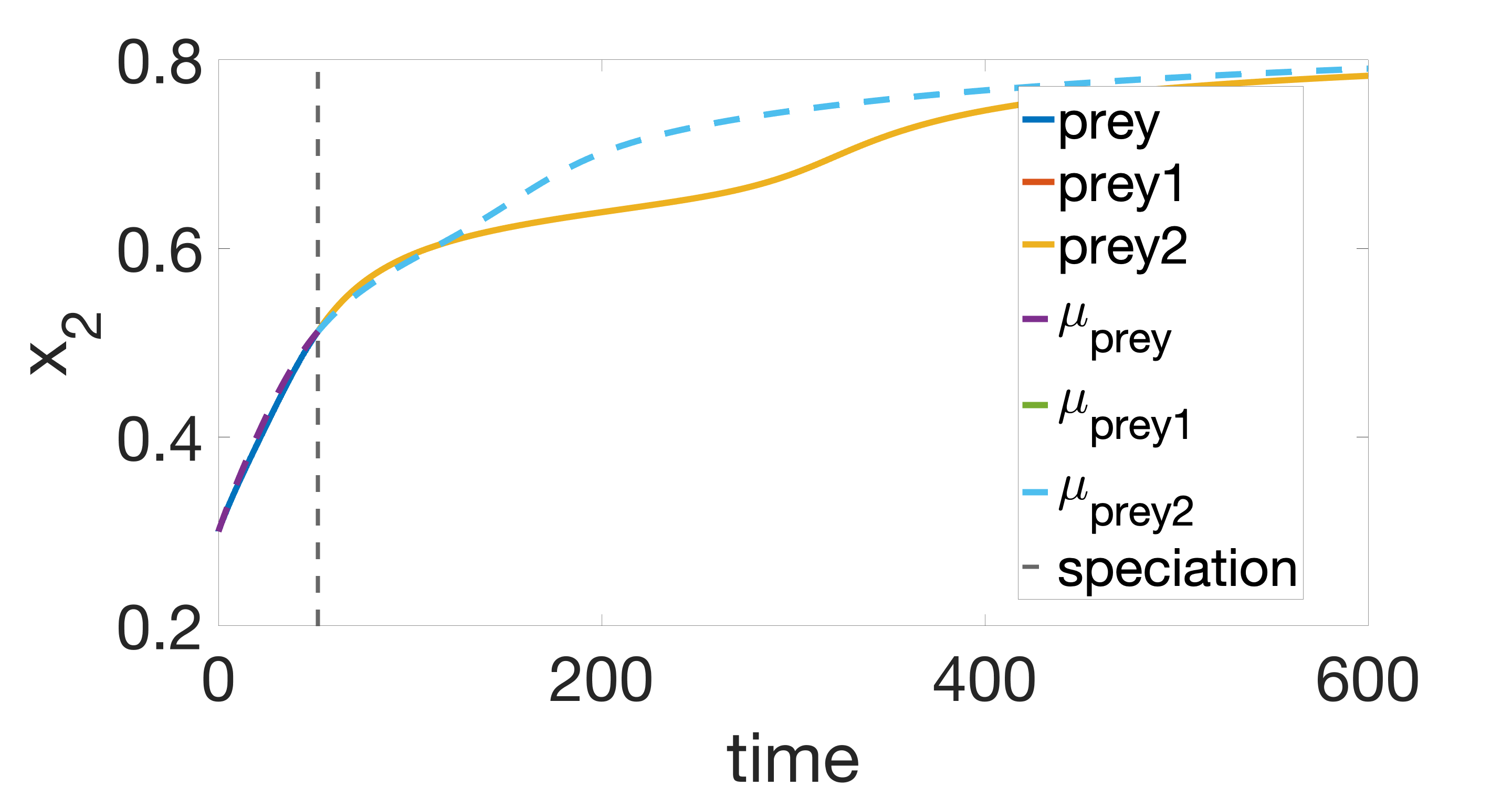}
  \caption{Heuristic method.}
\end{subfigure}  
  \caption{Prey mean traits coordinate as functions of time (second component). Curves for child-species are overlapping}
  \label{ex2:prey_x2}
\end{figure}


\begin{figure}[h]
  \centering
\begin{subfigure}{.45\textwidth}
  \centering
  \includegraphics[width=1.\linewidth]{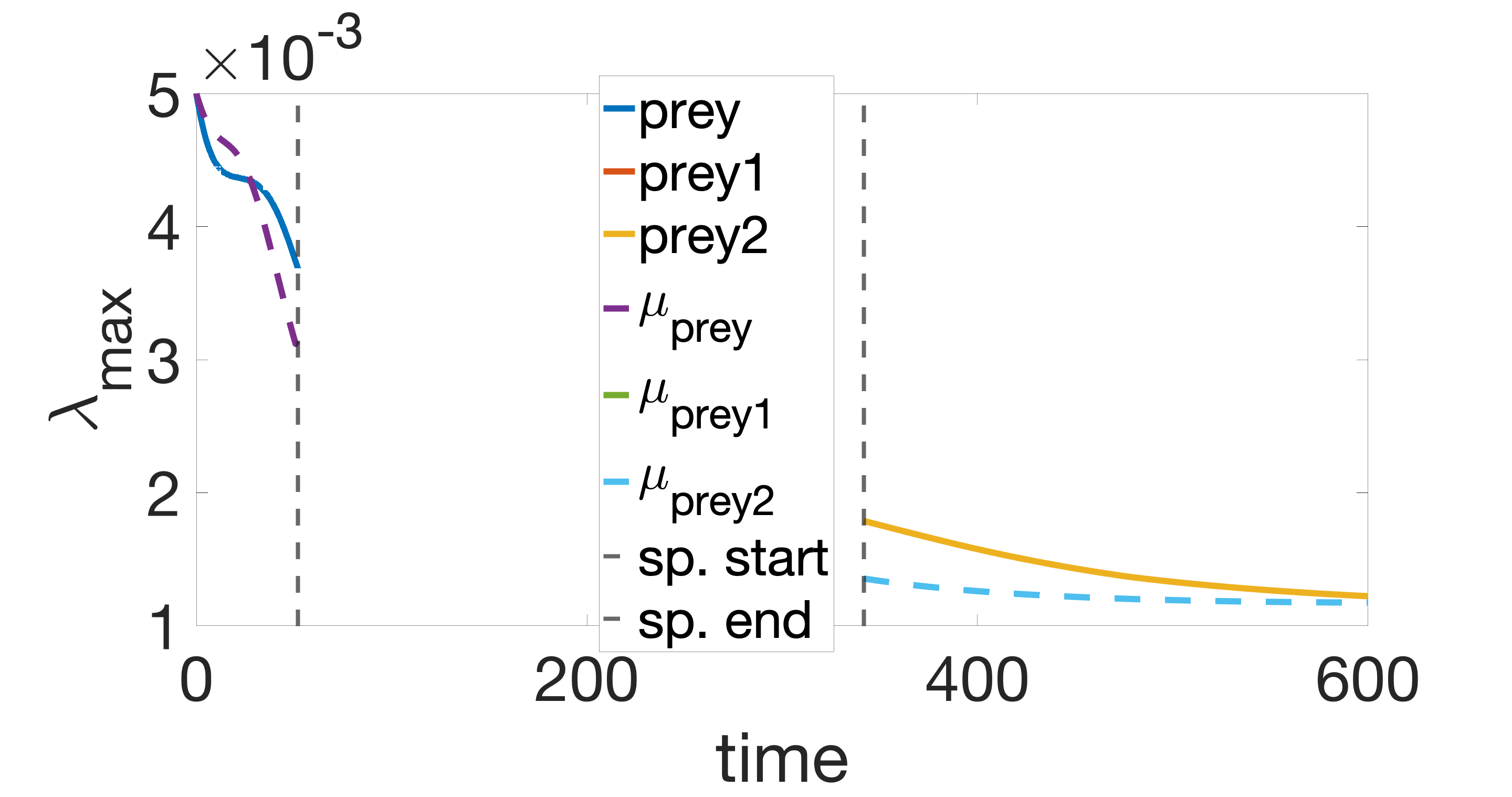}
  \caption{Multi-scale method.}
\end{subfigure}%
\hspace{.05\textwidth}
\begin{subfigure}{.45\textwidth}
  \centering
  \includegraphics[width=1.\linewidth]{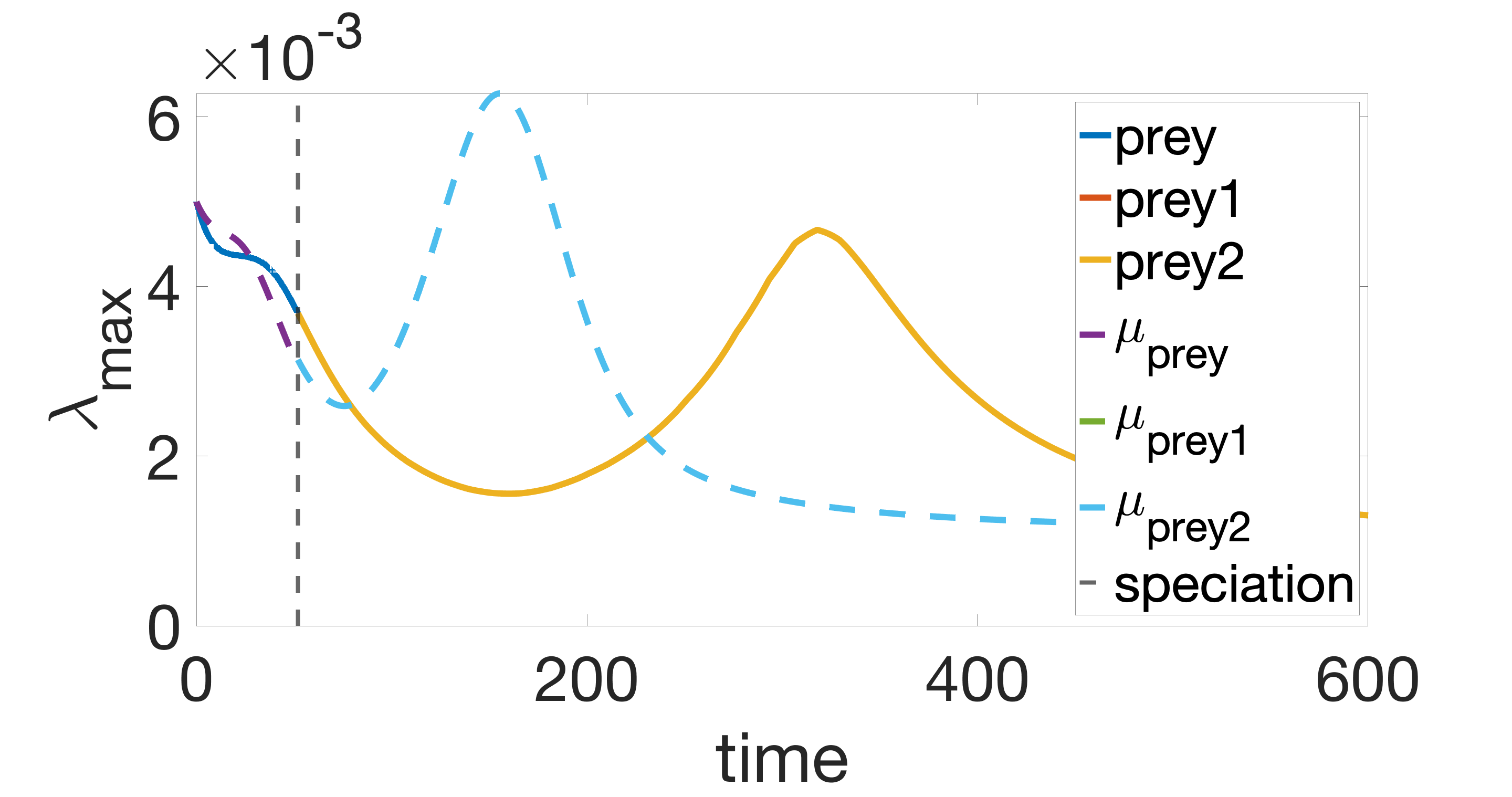}
  \caption{Heuristic method.}
\end{subfigure}  
  \caption{Prey maximum eigenvalue of trait covariance matrix as functions of time. Curves for child-species are overlapping.}
  \label{ex2:prey_variance}
\end{figure}


\begin{figure}[h]
  \centering
\begin{subfigure}{.45\textwidth}
  \centering
  \includegraphics[width=1.\linewidth]{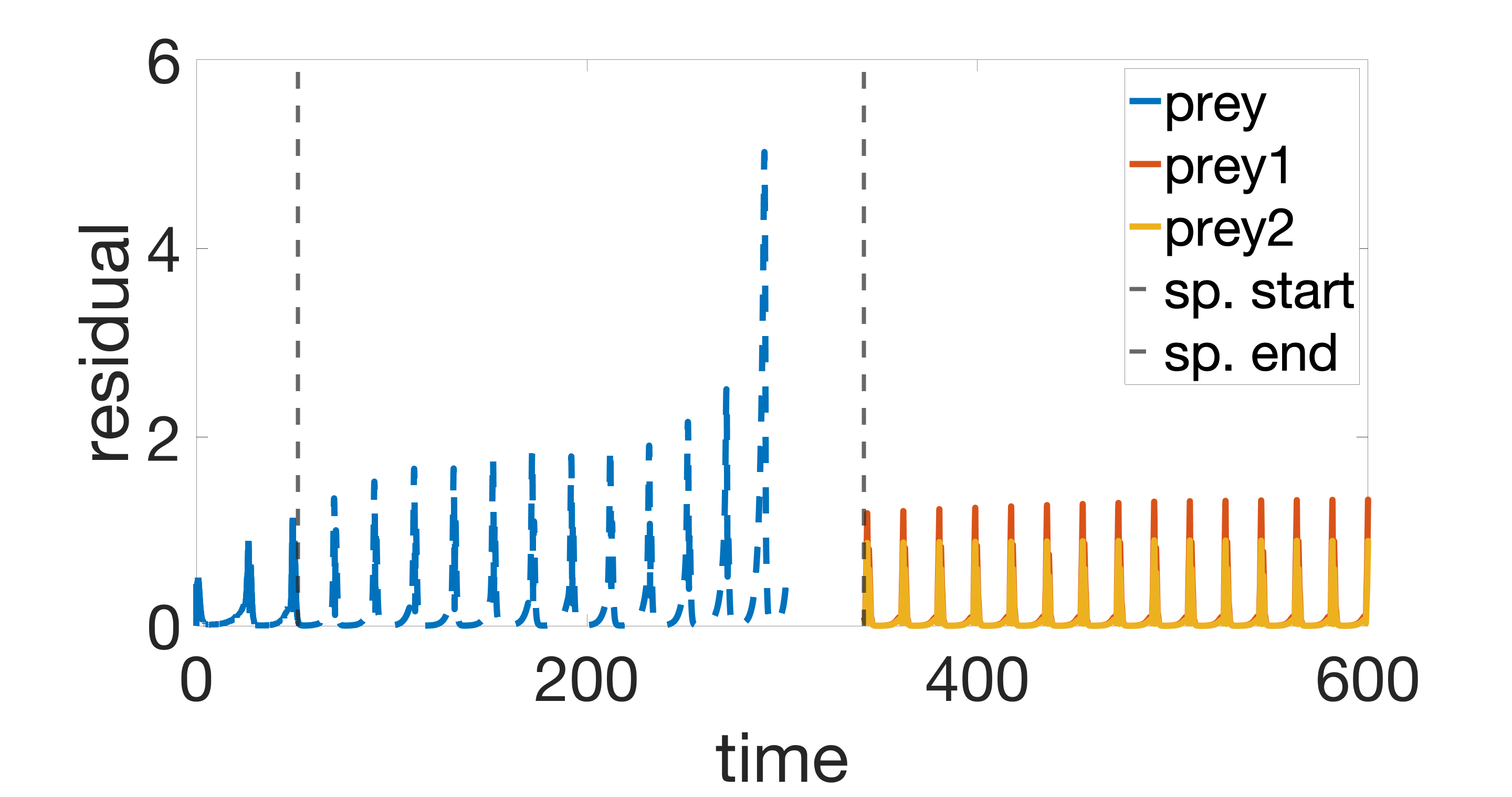}
  \caption{Multi-scale method.}
\end{subfigure}%
\hspace{.05\textwidth}
\begin{subfigure}{.45\textwidth}
  \centering
  \includegraphics[width=1.\linewidth]{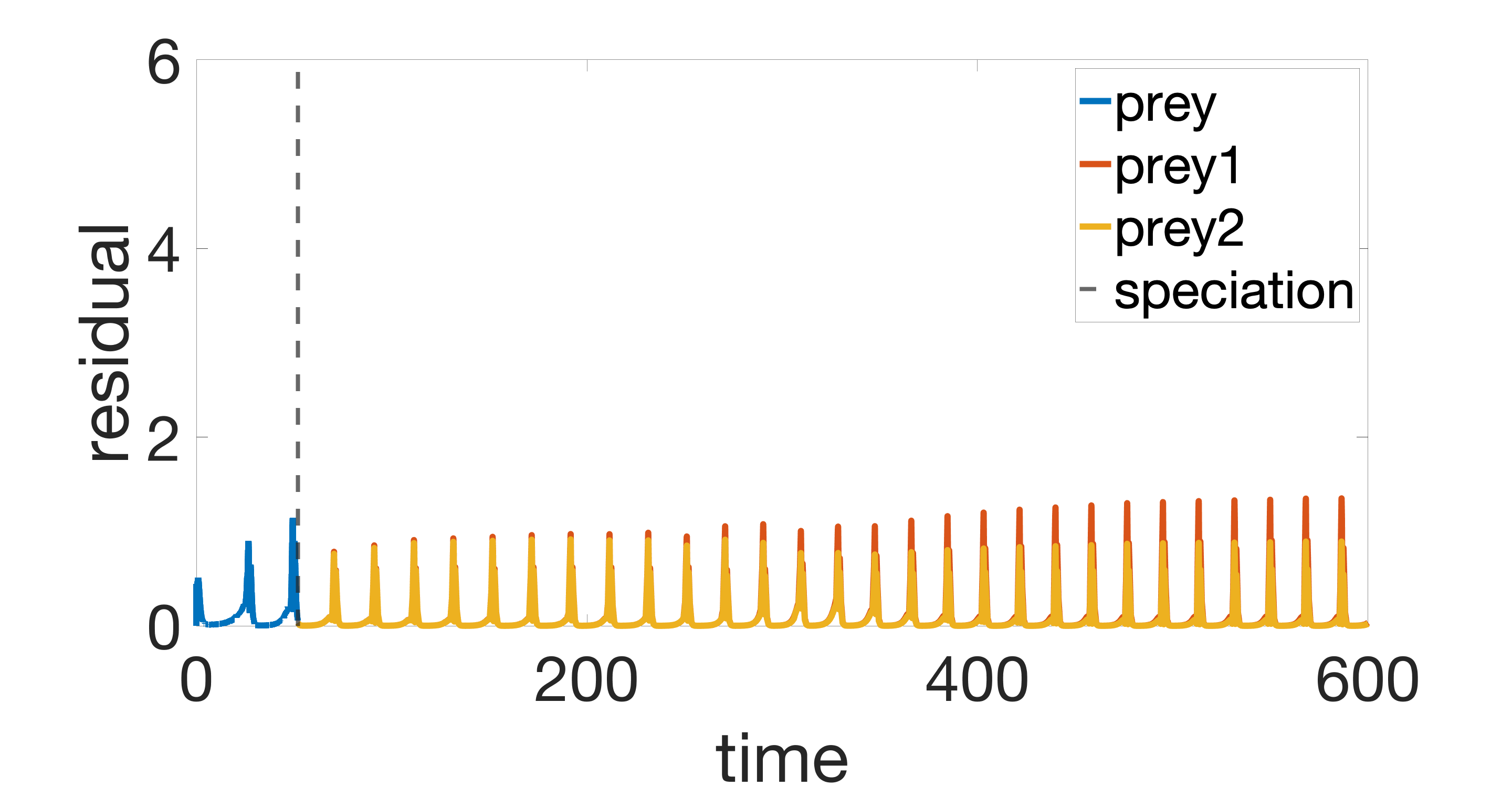}
  \caption{Heuristic method.}
\end{subfigure}  
  \caption{A posteriori modeling-remainder estimator, $\eta_{\textn{rem},i}^{k}$, for `prey'-species as functions of time.}
  \label{ex2:prey_residuals}
\end{figure}

\begin{figure}[h]
  \centering
\begin{subfigure}{.45\textwidth}
  \centering
  \includegraphics[width=1.\linewidth]{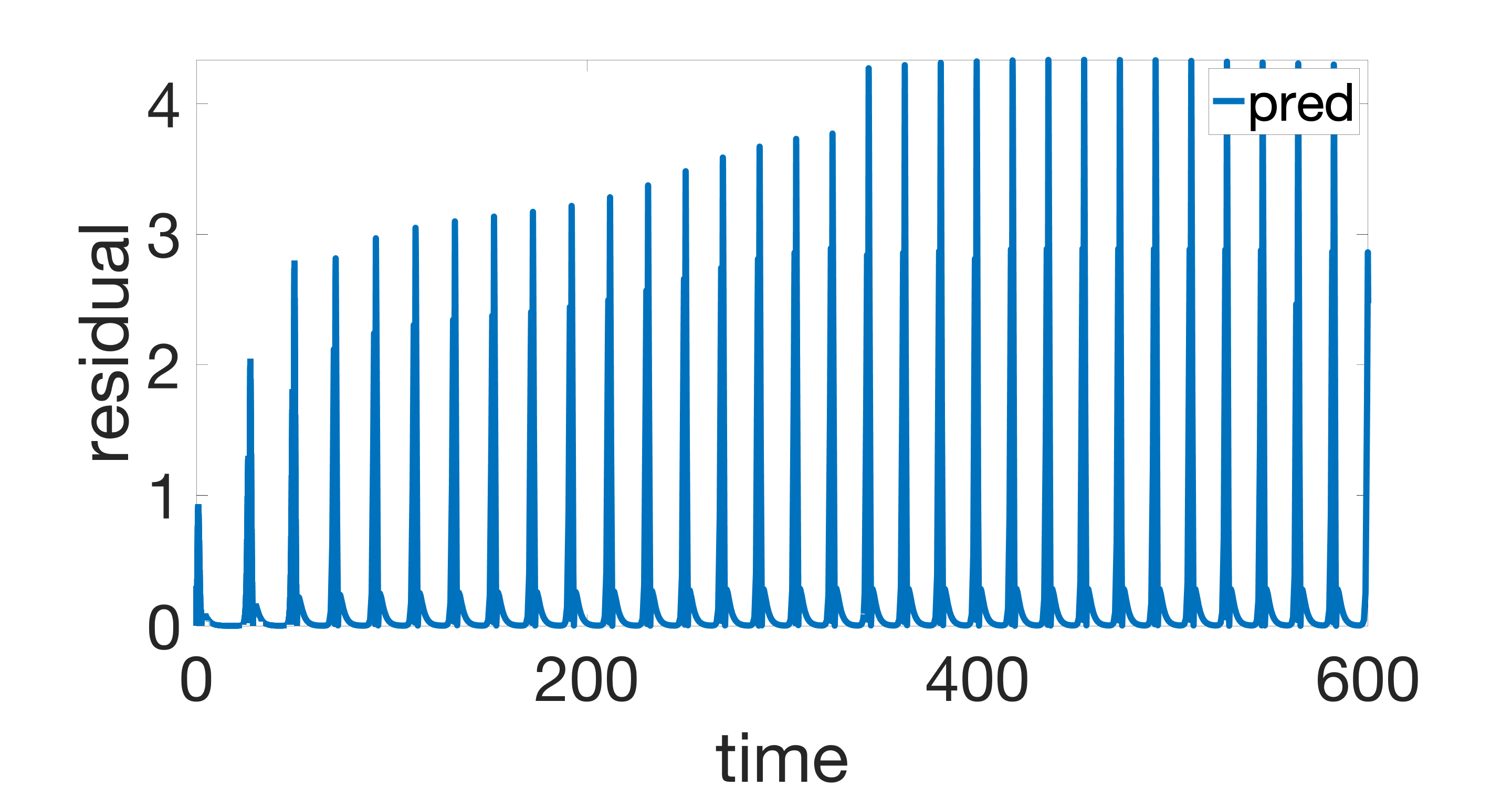}
  \caption{Multi-scale method.}
\end{subfigure}%
\hspace{.05\textwidth}
\begin{subfigure}{.45\textwidth}
  \centering
  \includegraphics[width=1.\linewidth]{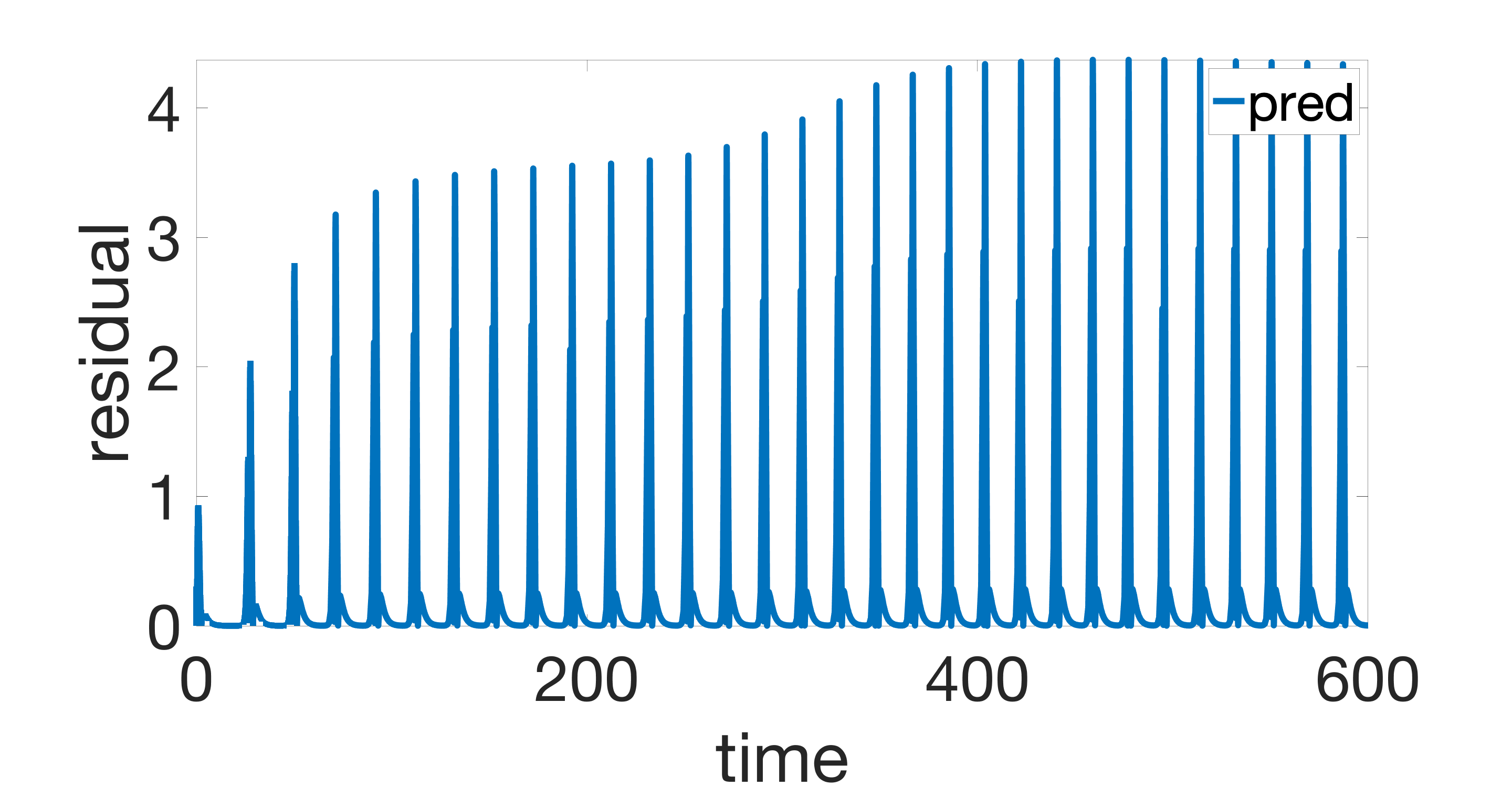}
  \caption{Heuristic method.}
\end{subfigure}  
  \caption{A posteriori modeling-remainder estimator, $\eta_{\textn{rem},i}^{k}$, for `predator'-species as a function of time.}
  \label{ex2:pred_residuals}
\end{figure}
\begin{remark}[Computation times]
The heuristic method will have a computation time approximately equal to that of solving only the species level model for the same time interval. On the other hand, the multi-scale method will have a computation time higher than that of solving the population level model for a time interval equal to the speciation interval (with the same grid size as the local density region) since the compression operator is implemented as an iteration procedure. Thus, the heuristic method will in general be orders of magnitude faster than the multi-scale method. 
\end{remark}

\section{Conclusions}
\label{sec:conclusions}
We have developed two strategies for modeling speciation \\ events within the context of species interaction models. The first, \emph{heuristic approach}, is based on splitting the diverging species according to the spectrum of the trait covariance matrix. The second, \emph{multi-scale approach}, is based on resolving the diverging species as a population density distribution using a fine-scale population level model for the duration of the speciation event (i.e., until `child'-species are sufficiently separated in trait space). Crucial to both these approaches is the connection between the species scale and the population scale, i.e., the ability to view the species either as an abundance-trait-covariance tuple, or as a population density distribution. This allows for defining a multi-scale framework in which these two scales are coupled, and to calculate the a posteriori error bound of the reconstructed macro-scale solution which then indicates the modeling error. We have also given conditions on the nonlinearities of the micro-scale model for which well-posedness of the time-discrete problem is guaranteed. Using explicit equilibrated flux and density reconstructions, we presented a posteriori error estimates for an error measure composed of an energy $H^1$-norm and a semi-metric in terms of a  residual monotone operator. In particular, the dual norm of the residual is found to be equal to the error between the exact and approximate solutions, again given conditions on the nonlinearities. Even when the theoretical conditions are not fulfilled, our framework provides a working algorithm in practice, as our second numerical example shows. Finally, regarding the heuristic and multi-scale methods, by comparison of the species parameters with the corresponding moments from the reference (global) PLM solution, it is clear that the multi-scale approach is superior to the heuristic approach. In fact, from our experiments it appears that a multi-scale approach to speciation is indeed required for eco-evolutionary modeling at the species level, in which speciation events are allowed. We propose that this multi-scale approach might serve as a productive way of integrating ecological processes and evolutionary processes.

\section*{Acknowledgments}
The research was funded in part by Norwegian Research Council project no. 263149.

\bibliographystyle{siamplain}
\bibliography{references}
\end{document}